\documentclass[a4paper,UKenglish]{lipics-v2016}
\usepackage{comment} 
\usepackage{microtype}

\usepackage[lined,linesnumbered,ruled]{algorithm2e}
\usepackage{tikz}
\usepackage{url}
\usepackage{xcolor}
\usepackage{hyperref,url,graphics,amsmath,amsfonts,amssymb,setspace,exscale,relsize,makeidx,enumitem}
\usepackage{lineno}

\newtheorem{observation}{\textbf{Observation}}[section]

\title{Fully dynamic maximal matching without $3$ length augmenting paths in O$(\sqrt{n})$ update time \footnote{The report is available at https://arxiv.org/pdf/1810.01073.pdf}}
\titlerunning{3/2 approximate MCM in O$(\sqrt{n})$ update time} 
\author[1]{Manas Jyoti Kashyop}
\author[1]{N.S. Narayanaswamy}

\affil{Department of Computer Science \& Engineering, \\ Indian Institute of Technology Madras,  
  Chennai, India \\
  \texttt{\{manasjk,swamy\}@cse.iitm.ac.in}}

\authorrunning{M. J. Kashyop and N. S. Narayanaswamy} 

\Copyright{Manas Jyoti Kashyop, N.S.Narayanaswamy}

\begin{document}

\maketitle
\begin{abstract}
We present a randomized algorithm to maintain a maximal matching without 3 length augmenting paths in the fully dynamic setting.  Consequently, we maintain a $3/2$ approximate maximum cardinality matching. Our algorithm takes expected amortized $O(\sqrt{n})$ time where $n$ is the number of vertices in the graph when the update sequence is generated by an oblivious adversary. Over any sequence of $t$ edge insertions and deletions presented by an oblivious adversary, the total update time of our algorithm is $O(t\sqrt{n})$ in expectation and $O(t\sqrt{n} + n \log n)$ with high probability. To the best of our knowledge, our algorithm is the first one to maintain an approximate matching in which all augmenting paths are of length at least $5$ in $o(\sqrt{m})$ update time.
\end{abstract}

\section{Introduction}
Dynamic graph algorithms is a vibrant area of research.  An \textit{update} operation on the graph is an insertion or a deletion of an edge or a vertex. The goal of a dynamic graph algorithm is to efficiently modify the solution after an update operation. 
A dynamic graph algorithm is said to be \textit{fully dynamic} if the update operation includes both insertion or deletion of edges (or vertices).
Dynamic algorithms for maintaining a maximal matching in a graph $G$ are well-studied and still present many research challenges. 
In particular, the approximate maximum cardinality matching (approximate MCM) problem is a very interesting question in the dynamic graph model.  It is very well known from \cite{DBLP:journals/siamcomp/HopcroftK73}	that for each $k \geq 1$, a maximal matching that does not have augmenting paths of length upto $2k-1$ is of size at least $\frac{k}{k+1}$ fraction of the maximum cardinality matching.
This theorem provides a natural way of addressing this problem in the dynamic setting.  However,  there are two conditional lower bounds on the complexity of maintaining a maximal matching by eliminating augmenting paths of length upto 5.  Assuming 3-sum hardness, Kopelowitz et al.\cite{DBLP:journals/corr/KopelowitzPP14} show that any algorithm that maintains a matching in which all the augmenting paths of length at most $5$ are removed requires an update time of $\Omega(m^{1/3} -\zeta)$ for any fixed $\zeta > 0$. Secondly, assuming the Online Matrix Vector Multiplication conjecture, Henzinger et al. \cite{DBLP:conf/stoc/HenzingerKNS15} show that any algorithm which maintains a maximal matching in which all the augmenting paths of length at most $5$ are removed  
requires an update time of $\Omega(m^{1/2} - \zeta)$ for $\zeta > 0$.   Consequently,  our aim in this paper is to understand the complexity of maintaining a maximal matching after eliminating augmenting paths of length 3.    In Table~[\ref{tab:approximate-matching}] we tabulate current results and our result.  The results in fourth row due to Neiman and Solomon \cite{DBLP:NeimanS13} output a maximal matching without 3 length augmenting paths in $O(\sqrt{m})$ time.    Bernstein et al. \cite{DBLP:BernsteinS} gave the first algorithm which achieve better than $2$ approximation in $o(
\sqrt{m})$ update time. They maintain a $3/2+\epsilon$ approximate matching in amortized $O(m^{1/4}/{\epsilon}^{2.5})$ update time. However their work does not provide the guarantee to maintain a maximal matching $M$ without any 3 length augmenting paths. 
We ask whether it is possible to remove all the augmenting paths of length up to $3$ in $o(m^{1/2})$ update time? We answer this question affirmatively. To the best of our knowledge, our algorithm is the first one to maintain an approximate matching in which all augmenting paths are of length at least $5$ in $o(\sqrt{m})$ amortized update time.  Our two main results are regarding the expected total update time, which is in Theorem \ref{thm : expected-total-update-time}, and regarding the worst case total update time with high probability which is  Theorem \ref{thm : worstcase-total-update-time}.

\noindent
\begin{table}[htb]
	\centering
	\caption{A comparison of some of the previous results and our result for the dynamic approximate MCM problem}		
	\label{tab:approximate-matching}
	\begin{tabular}{|p{4cm}|c|c|c| } \hline
		Reference & Approximation & Update time & Bound \\ \hline \hline
		Baswana et al. \cite{DBLP:BGS} & 2 & $O(\log n)$ and $O(\log n + (n \log^2 n)/t)$ & expected and w.h.p
		\\ \hline
		Solomon \cite{DBLP:conf/focs/Solomon16} & 2 & $O(1)$ & expected and w.h.p
		\\ \hline
		Peleg and Solomon (uniformly sparse graph) \cite{DBLP:conf/soda/PelegSolomon16} & $3/2 + \epsilon$ & $O(\alpha / \epsilon)$ & deterministic
		\\ \hline
		Gupta and Peng \cite{DBLP:conf/focs/GuptaPeng} & $1+\epsilon$ & $O(\sqrt{m}/ \epsilon^2)$ & deterministic
		\\ \hline
		Neiman and Solomon \cite{DBLP:NeimanS13} & $3/2$ & $O(\sqrt{m})$ & deterministic
		\\ \hline
		Bernstein and Stein \cite{DBLP:BernsteinS} & $3/2 + \epsilon$ & $O(m^{1/4} / \epsilon^{2.5})$ & deterministic
		\\ \hline 
		\textbf{This Paper}[Theorem~\ref{thm : expected-total-update-time}] & \textbf{3/2} & \textbf{O($\sqrt{n}$)} & \textbf{expected}
		\\ \hline
		\textbf{This Paper}[Theorem~\ref{thm : worstcase-total-update-time}] & \textbf{3/2} & \textbf{O($\sqrt{n} + (n \log n)/t$)} & \textbf{w.h.p}
		\\ \hline
	\end{tabular}
\end{table}
{\bf Paper Outline:}
In Section \ref{subsec:overview_baswana_algo} we present an overview of how our algorithm is different from and how it extends the algorithm of Baswana et. al \cite{DBLP:BGS}.
In Section \ref{algorithm} we describe the invariants maintained by our algorithms and the procedures in our algorithm. In Section~\ref{sec: Algorithm_correctness}, we present the correctness of our algorithm.  In Section \ref{epochanalysis}, we present the analysis of the expected total update time and the worst case total update time. 
\section{Outline of our Algorithm}
\label{subsec:overview_baswana_algo}
To achieve our main result we extend the algorithm to maintain a maximal matching in the first part of the work by Baswana et al. \cite{DBLP:BGS} (hereafter referred to as BGS).  We start by a brief and quick presentation of the ideas in BGS.  BGS introduced the notion of edge ownership by a vertex.  The number of edges owned by a vertex is at most its degree, and a vertex always searches for a mate among the edges that it owns.  For every vertex $u$ this information is maintained as a ownership list $O_u$.  For an edge $(u,v)$, if the edge is not owned by $u$ then it is owned by $v$. Further, they  maintain the invariant that if an edge $(u,v)$ is not owned by $u$, then it means that $v$ is matched. Therefore, for $u$ to look for a mate, it is sufficient for it to search for a mate among the edges it owns. 
BGS has two very important operations done by a vertex : scan its ownership list, and transfer an edge in its ownership list to the other vertex in the edge. The transfer of ownership is designed to ensure that if a vertex has to remain unmatched due to an update, then its ownership list is small (this does help a fast search for a mate subsequently).   Clearly, these operations become expensive when the ownership list becomes \textit{large}. BGS addressed this in the presence of an oblivious adversary, that is  an adversary who does not generate the dynamic update requests based on the run-time behaviour of the algorithm.  They addressed this by selecting a random mate from the ownership list whenever the ownership list of a matched vertex becomes \textit{large}. Specifically,  they maintain a maximal matching in which a vertex whose ownership list is large is matched.  The algorithm is implemented by organizing the vertices into two levels, namely level 0 and level 1.  A vertex in level 0 has a small ownership list, and a vertex with a large ownership list is in level 1 and is matched.  To analyze the expected amortized cost of maintaining a maximal matching, BGS first observes that the only updates which incur a cost greater than a constant are those that  delete an edge in the matching or those that insert an edge between unmatched vertices.  Further they classify the cost into two parts: that incurred by updates which involve vertices of low ownership (at level 0), and the cost incurred at vertices (at level 1) of high ownership.   The cost incurred at a vertex in level 0 is small, since the ownership list is small.  To analyze the expected amortized cost at level 1 vertices, BGS observe that the probability that an update by an oblivious adversary presents a randomly selected edge by the algorithm is at most the reciprocal of the size of the ownership list (which is large at level 1 vertices).   This ensures that the expected amortized cost of maintaining a maximal matching over $t$ updates is $O(t \cdot \text{ small } + n \cdot \frac{t}{\text{large}})$.  By considering a ownership of size $\sqrt{n}$ to be large, it follows that BGS maintain a maximal matching with expected amortized cost $O(\sqrt{n})$.
\subsubsection*{Removing 3 length augmenting paths-Our contribution}
Conceptually, we extend BGS by deterministically checking at the end of each update whether a vertex affected by the update is part of 3 length augmenting path.  For this our algorithm has to scan the whole neighbourhood of a matched vertex for an unmatched neighbour.  Pessimistically, it seems that an oblivious adversary could force any algorithm to spend a high cost eliminating 3 length augmenting paths after finding mate with lower expected amortized cost.   We observe by a careful analysis of our algorithm that this is not the case for eliminating 3 length augmenting paths.  Central to our approach is the concept of an epoch in BGS.  An epoch is a maximal runtime interval during which an edge is in the matching.  As in BGS we associate with the beginning and ending of an epoch, the running time of the procedures that insert an edge into the matching and remove an edge from the matching.  Apart from the procedures in BGS, we have the additional procedure to remove 3 length augmenting paths.  Therefore, a single edge update by the oblivious adversary can trigger a sequence of changes to the matching, and apriori it is not  clear how to analyze the length of this sequence.  

Our first contribution is that we ensure that this sequence is of constant length.  In other words, any update only triggers a constant number of procedure calls, some to ensure that the matching is maximal, some to ensure that 3 length augmenting paths are eliminated, and some to ensure that unmatched vertices have a small ownership list and a small neighborhood list.  To ensure that each update triggers only a constant number of procedure calls, we observe that having each edge owned exactly by one end point is a very useful property.  This is one feature which to us very interesting and startkly different from BGS where an edge at level 0 is owned by both its vertices. Intuitively, what we gain here is that when a vertex comes back down to level 0, it does not have to enter the ownership list of its neighbours, whose ownerships list may then become too big for them to be in level 0 thus triggering a chain of data structure modifications that we do not know how to bound.

The second crucial work that we have done is to set up the framework to prove that our updates terminate and correctly maintain all the invariants.   This proof is based on the design of the procedures called in each update. As per the design a procedure {\em immediately} repairs an invariant when it is violated before continuing to the other statements in the procedure.  The repair is done either by executing necessary statements (at most two assignment statements) immediately after the violation or by making an appropriate procedure call such that on return the invariant is satisfied.  Another property of the design is that when a 3 length augmenting path is eliminated it does not create another 3 length augmenting path.  The whole analysis is intricate and we hope to find a much cleaner argument.  

The third crucial idea is that  if an oblivious adversary has to force a high cost during the search of 3 length augmenting path involving one matching edge then the adversary would succeed at this task with probability which is the reciprocal of the large degree or should expect to  make many low cost updates prior to the high cost update.  We ensure this by matching a vertex of high degree to a randomly chosen neighbour in its ownership list.  This is achieved by having a randomized courterpart for each deterministic method that changes the matching.  This idea is already present in BGS where Random-Settle is the randomized counterpart of Naive-Settle. We extend this by having a procedure to remove 3 length augmenting paths, a procedure to raise the level of a vertex, and both have their randomized counterparts.  

To complete the analysis, we crucially observe that the expensive procedure calls made during runtime can be classified into one of the following two types:
\begin{itemize}
 \item Those expensive procedure calls that are associated with many low cost udpates.  
 \item The remaining expensive calls can be grouped into sets of constant size such that each set has a procedure call that matches a vertex to a randomly chosen neighbour in its ownership list.
\end{itemize}
Finally, like BGS we classify vertices into level 0 and level 1, but with rules that include the vertex degree also, apart from the size of the ownership list. Unlike BGS who use a Free array to keep track of whether a vertex is matched, we maintain a free neighbour list for each vertex. At first sight this operation is indeed an expensive operation, that an unmatched vertex must be maintained in the free neighbour list of all its neighbours.  However, the expected amortized cost is still controlled to be $O(\sqrt{n})$ by our algorithm.  
\section{Fully Dynamic Algorithm} 
\label{algorithm}
Let $G(V,E)$ be an undirected graph. Vertex set $V$ does not change throughout the algorithm and  $|V|=n$. The edge set $E$ changes during the course of  the algorithm due to the insert and delete operations. We start with an empty graph. At every update step we allow an insertion or the deletion of exactly one edge. Given a matching $\mathcal{M}$, an edge $(u,v)$ is said to be matched if $(u,v) \in \mathcal{M}$, otherwise edge $(u,v)$ is said to be unmatched. A vertex $u$ is said to be matched if there exist an edge $(u,v)$ such that $(u,v) \in \mathcal{M}$, otherwise $u$ is called free. If $(u,v) \in \mathcal{M}$ then $v$ is called \textit{mate} of $u$ and $u$ is called \textit{mate} of $v$.  For a vertex $u$ we write $mate(u)$ to denote the vertex to which $u$ is matched, and if $u$ is not matched we consider $mate(u)$ to be NULL.  Therefore, $u$ is free is $mate(u)$ is NULL.  Let $u-v-x-y$ be an augmenting path of length 3 where $u$ and $y$ are free vertices and the edge $(v,x)$ is matched.   
We call \textit{$u-v-x-y$ is a 3 length augmenting path involving $(v,x)$ as the matched edge}. Over a sequence of updates, we maintain a maximal matching $M$ of $G$ and ensure that the graph does not have an augmenting paths of length 3 with respect to $M$. Therefore, our matching $M$ is a $3/2$ approximate MCM.  
Similar to \cite{DBLP:BGS}, in our algorithm we partition the vertex set into two levels 0 and 1. $Level(v)$ denotes the level of a vertex $v$. For an edge $(u,v)$, $Level(u,v)$ = max$\{Level(u),Level(v)\}$. We also define \textit{ownership} for the edges with a slight difference from \cite{DBLP:BGS}. At every level, exactly one endpoint will own that edge. If both the endpoints of an edge have level 0 then the endpoint owning higher number of edges will own that edge. If level of one endpoint is $1$ and the other endpoint is $0$, then the endpoint with level $1$ owns that edge. If both the endpoints have level 1 then ownership is assigned to anyone of the endpoints arbitrarily. For every vertex we maintain two sets : $O_u$  and $N(u)$. $O_u$  denotes the set of edges that $u$ owns. $N(u)$ denotes the neighbours of $u$. 
\subsection{Invariants}
\label{subsec : Invariatnt}
Before an update step and after completion of an update step,
our algorithm maintains the following invariants :
\begin{enumerate}
	\item 
	\begin{enumerate}
	\label{invariant_1a}
	\item For each vertex $u$, if $u$ is a level 1 vertex, then $u$ is matched.
	\item
	\label{invariant_1b} For each vertex $u$, if $u$ is a free vertex, then $u$ is a level 0 vertex and all its neighbours are matched.
	\end{enumerate}
	\item 
	\label{invariant 2}
	For each vertex $u$, if $u$ is a level 0 vertex, then $|O_u| < \sqrt{n}$.
	\item 
	\label{invariant 3}
	For each vertex $u$, if $u$ is a level 0 matched vertex, then $deg(u) < \sqrt{n}$. Here $deg(u)$ is defined as $deg(u) = |N(u)|$.  Note that $deg(u)$ is the degree of $u$.  
	\item 
	\label{invariant 4}
	For each vertex $u$, if $u$ is matched vertex, then $u$ and $mate(u)$ are at the same level.
	\item 
	\label{invariant 5}
	For each vertex $u$, $u$ is not a matched vertex or a free vertex in a 3 length augmenting path with respect to $M$.  In other words, the graph does not have a 3 length augmenting path with respect to $M$. 
\end{enumerate} 
A vertex that violates any one of the invariants is said to be in a {\em dirty} state. A vertex that is not dirty is said to be in a {\em clean} state.    
So invariant 1 and 5 together implies that $M$ is a $3/2$ approximate MCM.
\begin{observation}
	\label{observation : aug-path}
	From Invariant~\ref{invariant_1b}, if $u$ is a free vertex with level 0 then all neighbours of $u$ are matched. Suppose at some update step, $u$ is matched to
	$u^{\prime}$ and prior to this update step $u$ was free and all its neighbours were matched. Then this update step will not result in a 3 length augmenting path which involves $(u,u^{\prime})$ as the matched edge.
\end{observation}
\subsection{Data Structures}
\label{sec:ds}
In this section we  describe the data structures used by our algorithm. Our data structures are similar to the data structures  used in \cite{DBLP:BGS} and \cite{DBLP:NeimanS13}.
\begin{itemize}
	\item Matching $M$ is maintained in a matrix of size $n \times n$. If edge $(u,v)$ is in the matching then entries $[u,v]$ and $[v,u]$ in the matrix are set to 1. Therefore, insertion into the matching and deletion from the matching takes $O(1)$ time.  Further, we maintain an array called $mate$ indexed by vertices where for each $u$, $mate(u)$ is either the vertex that $u$ is matched to, or if $u$ is unmatched it is NULL.  Therefore, whenever the matching is updated, $mate$ is also updated in constant time. Further, whether $u$ is free is checked in constant time by a query to $mate(u)$.  Therefore, in our procedure descriptions we check if $u$ is free without explicitly probing $mate(u)$ and describe the inserts and deletes into a matching by using set notation.
	\item For every vertex $v$, the set of neighbours of $v$, $N(v)$ is maintained as a dynamic hash table. Hence search, delete and insert can be performed in $O(1)$ time. Also the count $deg(v)$ = $|N(v)|$, which is the degree of $v$, is maintained for every vertex $v$. 
	\item For every vertex $v$ the set $O_v$ is stored as a dynamic hash table. Hence search, delete and insert can be performed in $O(1)$ time. 
	\item For every vertex $v$ a data structure $F(v)$ is maintained that stores the free neighbours of $v$. This data structure supports insert and delete in $O(1)$ time. In addition it supports the method \textbf{has-free}$(v)$ which returns TRUE if $v$ has a free neighbour in $O(1)$ time, and the method \textbf{get-free}$(v)$ which returns an arbitrary free neighbour of $v$ in $O(\sqrt{n})$ time. The data structure $F(v)$ is implemented as follows : for every vertex $v \in V$, a boolean array of size $n$ is maintained which indicates the free neighbours of $v$,  a counter array of size $\sqrt{n}$ in which the $j$-th element stores the number of free neighbours in the range $[\sqrt{n}.j+1, \sqrt{n}(j+1)]$, and a variable for total number of free neighbours. So insert,delete and \textbf{has-free}$(v)$ are clearly $O(1)$ operations. \textbf{get-free}$(v)$ requires a scan in the counter array to find a nonzero entry and then a scan in the boolean array in the appropriate range to find the free neighbour resulting in a $O(\sqrt{n})$ operation.   
\end{itemize}   
The following inline macros play a crucial role in a succinct presentation of the procedures.   Note that the inline macros are to be substituted by the associated code and should not be treated as function calls at run time.  \\
\textbf{Check-3-Aug-Path}$(u,v)$ : Here $u$ is a at level $0$ and $(u,v)$ is an edge where $v$ is a matched vertex. Let $y$ be the mate of $v$. If $y$ has a free neighbour $z$ ($z \neq u$) then $u-v-y-z$ is a $3$ length augmenting path, and if $y$ does not have a free neighbour then $z$ is considered to be NULL. The time required for insertion and deletion in $F(y)$ is $O(1)$ and \textit{get-free}$(y)$ takes $O(\sqrt{n})$ time. Therefore, the worst case time taken by this macro is $O(\sqrt{n})$.  Whenever this macro is executed, if a non NULL $z$ is returned, it will also hold that $z$ does not have a free neighbour (proved in Lemma \ref{lem:safez}).  Consequently, from Observation \ref{observation : aug-path}, it follows that after modifying the matching to remove the 3 length augmenting path, $u$ and $z$ will not be part of another 3 length augmenting path.  \\
\textbf{Transfer-Ownership-From}$(u)$ : This macro is executed when $u$ changes its level from $1$ to $0$. For every edge $(u,w) \in O_u$, if  the level of $w$ is $1$, then $(u,w)$ is deleted from $O_u$ and inserted to $O_w$. This maintains the invariant that an edge is always owned by the endpoint at higher level.  Each deletion from $O_u$ and insertion into $O_w$ takes $O(1)$ time. Therefore, the total time taken by this macro is $O(deg(u))$.\\
\textbf{Transfer-Ownership-To}$(u)$ : This macro is executed  when vertex $u$ changes its level from $0$ to $1$. For every edge $(u,w) \in O_w$, if the level of $w$ is $0$, then $(u,w)$ is deleted from $O_w$ and inserted to $O_u$. This maintains the invariant that an edge is always owned by the endpoint at higher level. Each deletion from $O_w$ and insertion  into $O_u$ takes $O(1)$ time. Therefore, the total time taken by this macro is $O(deg(u))$.\\
\textbf{Take-Ownership}$(u)$ : This macro is executed for a matched vertex $u$ at level $0$ whenever $deg(u) \geq \sqrt{n}$. For every edge $(u,w) \in N(u)$, if $(u,w) \in O_w$, then $(u,w)$ is deleted from $O_w$ and inserted to $O_u$. After this transfer, $|O_u| \geq \sqrt{n}$. As the name suggests, $u$ has taken ownership of all the edges incident on it.  Each deletion  from $O_w$ and insertion into $O_u$ takes $O(1)$ time. Therefore, the total time taken by this macro is $O(deg(u))$.\\
\textbf{Insert-To-F-List}$(u)$ : This macro inserts $u$ into the free neighbour list of all its neighbours. For every $(u,w) \in N(u)$, $u$ is inserted in $F(w)$ in $O(1)$ time. Therefore, the total time taken by this macro is $O(deg(u))$.\\
\textbf{Delete-From-F-List}$(u)$ :  This macro deletes $u$ from the free neighbour list of all its neighbours. For every $(u,w) \in N(u)$, $u$ is deleted from $F(w)$ in $O(1)$ time. Therefore, the total time taken by this macro is $O(deg(u))$.
\subsection{Description of the Procedures}
\label{subsec : Description_Procedures}
The detailed description of the procedures along with the proofs of correctness and run-time analysis is also presented in this section.  We give a summarized description of these procedures in Table~\ref{table : Procedure_Description}.  Since these procedures call each other in  Table~\ref{table : Procedure_Calls} we summarize the  calling procedures and the conditions at the time of call.  In this table we also point out the invariants that are satisfied at the end of each procedure call.  The pseudocode for each of the procedures is presented in Section~\ref{sec : Appendix}.

The main update functions are insert and delete which are in Section \ref{subsec : insertion} and Section \ref{subsec : deletion}.   The update functions call different procedures to ensure that all the invariants are satisfied at the end of the update.  These procedures, {\em first} on entry into the body of the procedure, manipulate the data structures described in Section \ref{sec:ds} to address a violated invariant, and on each manipulation check for a violated invariant. if a violated invariant is found, an appropriate procedure to fix it is called.  The manipulation of the data structures and the check for violated invariants are done by statements in the procedure or the inline macros described in Section \ref{sec:ds}.  

The names of the procedures called by the updates are also chosen in a very suggestive way, and also as an extension of the names in \cite{DBLP:BGS}: A {\em settle} function finds the mate of a vertex, a {\em Raise} procedures changes the level of a vertex from level 0 to level 1, a {\em Fix} procedures removes 3 length augmenting paths, and the {\em Handle} procedure deals with the delete of an edge in the matching at level 1.  The {\em Settle}, {\em Raise}, and {\em Fix} procedures have a randomized version and a deterministic version, and the appropriate version is called during an update based on the value of the the boolean variable {\em flag}.  At the beginning of an update {\em flag} is initialized to 0.  It is set to 1 once a level 1 vertex is assigned a random mate during the update. Once {\em flag} is 1, the deterministic variants of {\em Raise}, {\em Fix}, and {\em Settle} are called.\\
{\bf Generic Outline of the Procedures:}  
Each procedure described below is a specification of a set of actions to be taken for corresponding a set of cases.  Each procedure ensures that the {\em violated} invariant is {\em satisfied} on entry into the procedure by suitably modifying a combination of the following 3 data items associated with a vertex: its mate, its ownership, its level number.  Then based on whether a new invariant is violated, it executes at most one case and associated sub-cases. Finally,  where appropriate newly created 3 length augmenting paths are fixed before returning to the calling procedure.  As will be clear, each case is handled by a sequence of procedure calls.   Finally, each procedure executes necessary macros to update the ownership of a vertex, to identify a 3 length augmenting path, if any, and to update free neighbour lists.  

A crucial point is that the procedures are designed in such a way that no procedures calls are made in iterations that involve the whole neighbourhood of a vertex.  Any iterative work is performed inside an appropriate macro.  The reason is that during our analysis we associate we each procedure call a modification to the matching, and associate the total time of the procedure execution to the modification that happens to the matching.  
\begin{table}[htbp!]
	\centering
	\begin{tabular}{|m{3cm}|m{9cm}|m{2.5cm}|}
		\hline
		\textbf{Procedure} & \centering \textbf{Description} & \textbf{Computation time} \\
		\hline
		\multirow{4}{3cm}{Naive-Settle-Augmented(u,flag)}
		& Case $1 :$ If $w$ is in $F(u)$ and $deg(u) \geq \sqrt{n}$. If $flag$ is $1$ then Deterministic-Raise-Level-To-1$(x)$ is called. Else Randomised-Raise-Level-To-1$(x)$ is called. & $O(n)$\\
		\cline{2-3}
		& Case $2 :$ If $w$ is in $F(u)$ , $deg(u) < \sqrt{n}$ and $deg(w) \geq \sqrt{n}$. If $flag$ is $1$ then Deterministic-Raise-Level-To-1$(w)$ is called. Else Randomised-Raise-Level-To-1$(w)$ is called. & $O(n)$\\
		\cline{2-3}
		& Case $3 :$ If $w$ is in $F(u)$, $deg(u) < \sqrt{n}$ and $deg(w) < \sqrt{n}$. Edge $(u,w)$ is included in matching & $O(\sqrt{n})$\\
		\cline{2-3}
		& Case $4 :$ If the list $F(u)$ is empty, then to check for 3 length augmenting paths, procedure Fix-3-Aug-Path is called if $flag$ is $0$ and procedure Fix-3-Aug-Path-D is called otherwise. & $O(n)$\\
		\hline
		\multirow{2}{3cm}{Random-Settle-Augmented(u)}
		& Case $1 :$ Selects an edge, say $(u,y)$, uniformly at random from $O_u$. If $y$ was matched then it returns the previous mate of $y$. 
		If $u$ has no free neighbour then procedure stops. & $O(deg(u) + deg(y))$\\
		\cline{2-3}
		& Case $2 :$ Selects an edge, say $(u,y)$, uniformly at random from $O_u$. If $y$ was matched then it returns the previous mate of $y$.  If $w$ is in $F(u)$ then procedure Fix-3-Aug-Path-D$(w,u)$ is called.& $O(n)$	\\	
		\hline
		Deterministic-Raise-Level-To-1$(u)$
		& 
		Vertex $u$ is matched, $deg(u) \geq \sqrt{n}$ and level of $u$ is $0$. Let $v$ is mate of $u$. Procedures Take-Ownership$(u)$ and Transfer-Ownership-To$(v)$ are called. Levels of $u$ and $v$ are changed to $1$.& $O(deg(u) + deg(v))$\\
		\hline
		Randomised-Raise-Level-To-1$(u)$ 
		& Vertex $u$ is matched, $deg(u) \geq \sqrt{n}$ and level of $u$ is $0$. Let $v$ is mate of $u$. Edge $(u,v)$ is removed from the matching. Procedure Take-Ownership$(u)$ is called.  Procedure Random-Settle-Augmented$(u)$ is called. If it returns a vertex $x$ and level of $x$ is $1$ then Handle-Delete-Level1(x,1) is called. Otherwise Naive-Settle-Augmented(x,1) is called. If $v$ is free then Naive-Settle-Augmented(v,1) is called. & $O(n)$\\
		\hline
		Fix-3-Aug-Path-D$(u,v,y,z)$ & $u$ is free, level of $u$ is $0$ and $v$ is matched. Let $y$ be the mate of $v$. If $z$ ($z \neq u$) is in $F(y)$, then edge $(v,y)$ is removed from matching and edges $(u,v)$ and $(y,z)$ are included in the matching. Levels of $u$, $z$ ($v$,$y$ if required) are changed to $1$ &  $O(deg(u)+deg(z))$ \\
		\hline
		\multirow{5}{3cm}{{Fix-3-Aug-Path$(u,v,y,z)$:} \emph{Initial Condition:}
		 $u$ is free, level of $u$ is $0$, $v$ is matched and level of $v$ is $1$ or $0$. Let $y$ be the mate of $v$ and $z$ ($z \neq u$) is in $F(y)$.
		}
		& \emph{Common step :} $(v,y)$ is removed from matching and $(u,v)$ and $(y,z)$ are included in the matching. & \\
		\cline{2-3}
		& Case $1: $  If $deg(u) \geq \sqrt{n}$ and level of $v$ is $1$, then procedure Random-Raise-Level-To-1$(u)$ is called. level of $z$ is changed to $1$& $O(n)$ \\
		\cline{2-3}
		& Case $2: $ If $deg(u) < \sqrt{n}$, $deg(z) \geq \sqrt{n}$ and level of $v$ is $1$ then procedure Random-Raise-Level-To-1$(z)$ is called. level of $u$ is changed to $1$ & $O(n)$ \\
		\cline{2-3}
		& Case $3: $ If $deg(u) < \sqrt{n}$, $deg(z) < \sqrt{n}$ and level of $v$ is $1$ then, levels of $u$ and $z$ are changed to $1$ & $O(\sqrt{n})$ \\
		\cline{2-3}
		& Case $4:$ Level of $v$ is $0$. If $deg(u) \geq \sqrt{n}$ then procedure Random-Raise-Level-To-1$(u)$ is called. If $deg(z) \geq \sqrt{n}$ then procedure Random-Raise-Level-To-1$(z)$ is called. & $O(n)$ \\
		\cline{2-3}
		& Case $5: $ Level of $v$ is $0$. If $deg(u) < \sqrt{n}$, $deg(z) < \sqrt{n}$ then, no further processing required. & $O(\sqrt{n})$ \\
		\hline
		\multirow{3}{3cm}{Handle-Delete-Level1$(u,flag)$} & \emph{Common step :} $u$ is free, level of $u$ is $1$. Procedure Transfer-Ownership-From$(u)$ is called. Then :
		& \\
		\cline{2-3}
		& Case $1: $ If $|O_u| \geq \sqrt{n}$, then procedure Random-Settle-Augmented$(u)$ is called. & $O(n)$ \\
		\cline{2-3}
		& Case $2: $ If $|O_u| < \sqrt{n}$, then Level of vertex $u$ is changed to $0$. Procedure Naive-Settle-Augmented$(u,flag)$ is called. & $O(n)$ \\
		\hline	
	\end{tabular}
	\caption{Procedure Description table. Note that the value of $flag$ distinguishes between a call to a randomised method ($flag = 0$) and the corresponding deterministic method ($flag = 1$)}
	\label{table : Procedure_Description}
\end{table}
\begin{table}[htbp!]
	\centering
	\begin{tabular}{|m{3cm}|m{4cm}|m{4cm}|m{2.5cm}|}
		\hline
		\textbf{Called Procedure} & \textbf{Calling Procedure} & \centering \textbf{Conditions at the time of call} & \textbf{Invariant satisfied by Called Procedure}\\
		\hline
		\multirow{3}{3cm}{Naive-Settle-Augmented(x,flag)}
		& Randomised-Raise-Level-To-1$(u)$: $u$ is randomly matched to $u^{\prime}$ and $x$ = $mate(u^{\prime})$ & x is free , Level of x is $0$ and flag is $1$ & satisfies Invariant~\ref{invariant_1b} for $x$.\\
		\cline{2-4}
		& Randomised-Raise-Level-To-1$(y)$: $y$ is randomly matched to $y^{\prime}$ and $x$ is previous mate of $y$. & x is free , Level of x is $0$ and flag is $1$ & satisfies Invariant~\ref{invariant_1b} for $x$.\\
		\cline{2-4}
		& Handle-Delete-Level1$(x,flag)$ & $x$ is free, Level of $x$ is $0$ and flag is $0$ or $1$ & satisfies Invariant~\ref{invariant_1b} for $x$.\\
		\cline{2-4}
		& Handle-Delete-Level1$(u,flag)$ : $u$ is randomly matched to $u^{\prime}$ and $x$ = $mate(u^{\prime})$ & $x$ is free, Level of $x$ is $0$ and flag is $1$  & satisfies Invariant~\ref{invariant_1b} for $x$.\\
		\hline
		\multirow{3}{3cm}{Random-Settle-Augmented(x)}
		& Handle-Insert-Level0$(x,v)$ & $x$ is free, Level of $x$ is 0 and $|O_x| = \sqrt{n}$ & satisfies Invariant~\ref{invariant 2} for $x$.
		\\
		\cline{2-4}
		& Handle-Delete-Level1$(x,flag)$ : $flag$ is $0$ or $1$ & $x$ is free, Level of $x$ is $0$ and $|O_x| \geq \sqrt{n}$ & satisfies Invariant~\ref{invariant 2} for $x$.
		\\
		\cline{2-4}
		& Randomised-Raise-Level-To-1$(x)$ & $x$ is free, Level of $x$ is 0 and $|O_x| \geq \sqrt{n}$ & satisfies Invariant~\ref{invariant 2} for $x$.
		\\
		\hline
		Deterministic-Raise-Level-To-1(x)
		& Naive-Settle-Augmented$(x,flag=1)$ & $x$ is matched, Level of $x$ is 0 and $deg(x) \geq \sqrt{n}$ & satisfies Invariant~\ref{invariant 3} for $x$.
		\\
		\hline
		\multirow{2}{3cm}{Randomised-Raise-Level-To-1(x)}
		& Naive-Settle-Augmented$(x,flag=0)$ & $x$ is matched, Level of $x$ is 0 and $deg(x) \geq \sqrt{n}$ & satisfies Invariant~\ref{invariant 3} for $x$.
		\\
		\cline{2-4}
		& Fix-3-Aug-Path$(x,v,v^{\prime},z)$ : $v^{\prime}$ = $mate(v)$ and $z \in F(v^{\prime})$ & $x$ is matched, Level of $x$ is 0 and $deg(x) \geq \sqrt{n}$ & satisfies Invariant~\ref{invariant 3} for $x$.
		\\
		\hline
		\multirow{2}{3cm}{Fix-3-Aug-Path-D$(x,y,v,z)$} & Random-Settle-Augmented$(y)$ : $x$ is in $F(y)$ & $x$ is free, Level of $x$ is 0, $y$ is matched, $v$ is mate of $y$, Level of $y$ is 1 and $z \in F(v)$ & satisfies Invariant~\ref{invariant 5} for $x$.
		\\
		\cline{2-4}
		& Naive-Settle-Augmented$(x,1)$ : $y \in N(x)$ , $v$ = $mate(y)$ and $z$ is in $F(v)$, $flag$ is $1$ & $x$ is free, Level of $x$ is 0, $y$ is matched & satisfies Invariant~\ref{invariant 5} for $x$.
		\\
		\hline
		\multirow{3}{3cm}{Fix-3-Aug-Path(x,y,v,z)}
		& Naive-Settle-Augmented$(x,0)$: $y \in N(x)$ , $v$ = $mate(y)$ and $z$ is in $F(v)$, $flag$ is $0$ & $x$ is free, Level of $x$ is 0  and $x$ does not have a free neighbour & satisfies Invariant~\ref{invariant 5} for $x$.
		\\
		\cline{2-4}
		& Handle-Insert-Level0$(x,y)$ & $x$ is free, Level of $x$ is 0 and $y$ is matched & satisfies Invariant~\ref{invariant 5} for $x$.
		\\
		\hline
		Handle-Delete-Level1(x,flag)
		& Randomised-Raise-Level-To-1$(u)$ : $u$ is randomly matched to $u^{\prime}$ and $x$=$mate(u^{\prime})$ & $x$ is free, Level of $x$ is 1 and flag is $1$ & satisfies Invariant~\ref{invariant_1a}a for $x$.
		\\
		\hline
	\end{tabular}
	\caption{Procedure call table.  For every deterministic procedure call we have a corresponding randomised counterpart.  Columns 2 and 3 describe the algorithm state before the procedure call. Column 4 lists the invariant {\em satisfied} by the called procedure prior to making any other function call. Note that the value of $flag$ distinguishes between a call to a randomised method ($flag = 0$) and the corresponding deterministic method ($flag = 1$)}
	\label{table : Procedure_Calls}
\end{table}
\subsubsection{Naive-Settle-Augmented$(u,flag)$} 
\label{subsubsec : Naive-Settle-Augmented}
This procedure receives a free vertex $u$ with level 0 and a $flag$ as input. 
As described in Algorithm~\ref{alg:Naive-Settle-Augmented},
the procedure \textit{Naive-Settle-Augmented}$(u,flag)$ works as follows :\\
It checks for a free neighbour of $u$ by executing the macro \textit{Has-free}$(u)$. The following two cases and the sub-cases describes the remaining processing.
\begin{enumerate}
	\item  Vertex $u$ has a free neighbour and let $w$ be the free neighbour returned by \textit{get-free}$(u)$. Therefore, $u$ and $w$ violates Invariant~\ref{invariant_1b}. Edge $(u,w)$ is included in matching $M$, thus Invariant~\ref{invariant_1b} is satisfied for both $u$ and $w$. 
	\begin{enumerate}
		\item If $deg(u) \geq \sqrt{n}$ then vertex $u$ violates Invariant~\ref{invariant 3}. If $flag$ is 1, then procedure \textit{Deterministic-Raise-Level-To-1}$(u)$ is called. Otherwise, \textit{Randomised-Raise-Level-To-1}$(u)$ is called. Procedure \textit{Deterministic-Raise-Level-To-1}$(u)$[see Section~\ref{subsubsec : Deterministic-Raise-Level-To-1}] and \textit{Randomised-Raise-Level-To-1}$(u)$ [see Section~\ref{subsubsection : Randomised-Raise-Level-To-1}] makes $u$ matched and changes its level to $1$ and hence satisfies Invariant~\ref{invariant 3} for vertex $u$. At the time of call to this procedure  $w$ was a free vertex.
		Therefore, if $u$ gets matched to a vertex other than $w$ after call to procedure \textit{Randomised-Raise-Level-To-1}$(u)$ then $w$ will remain a free vertex as it was during entry into this function.
		\item If $deg(u) < \sqrt{n}$ and $deg(w) \geq \sqrt{n}$ then vertex $w$ violates Invariant~\ref{invariant 3}. 
		If $flag$ is 1 then procedure \textit{Deterministic-Raise-Level-To-1}$(w)$ is called. Otherwise \textit{Randomised-Raise-Level-To-1}$(w)$ is called. As explained in the previous case, whichever of these two calls are made, it satisfies Invariant~\ref{invariant 3} for $w$. At end of call to procedure \textit{Randomised-Raise-Level-To-1}$(w)$, if $u$ remains free, then $u$ may violate Invariant~\ref{invariant_1b} or Invariant~\ref{invariant 5} and hence procedure \textit{Naive-Settle-Augmented}$(u,1)$ is called.
		\item If $deg(u) < \sqrt{n}$ and $deg(w) < \sqrt{n}$ then no further processing is required. Further, $u$ and $w$ are now matched vertices in level 0. 
	\end{enumerate}
	Let $w$ be the \textit{mate} of $u$. Vertices $u$ and $w$ are removed from free neighbour lists of all their neighbours. The control returns to the calling procedure as there are no new violated invariants.
	\item Vertex $u$ does not have a free neighbour. If $u$ is part of a 3 length augmenting path then $u$ violates Invariant~\ref{invariant 5}. To check this, the macro \textit{Check-3-Aug-Path}$(u,x)$ is executed for each $x \in N(u)$ till it identifies an $x$ for which $z$ is not NULL.  If such a $z$ is found then,
	\begin{enumerate}
		\item If $flag$ is $0$, then
		procedure \textit{Fix-3-Aug-Path}$(u,x,mate(x),z)$ is called. Procedure \textit{Fix-3-Aug-Path}$(u,x,mate(x),z)$ [see Section~\ref{subsubsec : Fix-3-Aug-Path}] satisfies Invariant~\ref{invariant 5} for vertices $(u,x,mate(x),z)$, due to Observation \ref{observation : aug-path}. 
		\item Else If $flag$ is $1$, then procedure \textit{Fix-3-Aug-Path-D}$(u,x,mate(x),z)$ is called. Procedure \textit{Fix-3-Aug-Path-D}$(u,x,mate(x),z)$ [see Section~\ref{susubsec : Fix-3-Aug-Path-D}] satisfies Invariant~\ref{invariant 5} for vertices $(u,x,mate(x),z)$, due to Observation \ref{observation : aug-path}.
	\end{enumerate} 
	If $z$ is NULL for each $x \in N(u)$, then it means that $u$ is not part of a $3$ length augmenting path, and then $u$ is added to the free neighbour list of each of its neighbours. The control returns to the calling procedure as there are no new violated invariants.
\end{enumerate}
\subsubsection{Random-Settle-Augmented$(u)$}
\label{subsubsec : Random-Settle-Augmented}
This procedure is invoked at a free vertex $u$ at level $0$ and $|O_u| \geq \sqrt{n}$. Therefore, $u$ violates Invariant~\ref{invariant 2}.
As described in Algorithm~\ref{alg:Random-Settle-Augmented}, the procedure works as follows : The procedure selects an edge uniformly at random from $O_u$. Let $(u,y)$ be the randomly selected edge. The macro \textit{Transfer-Ownership-To}$(y)$ is executed to ensure that $y$ owns all those edges incident on it whose other end is a level 0 vertex.  
If $y$ was matched, let $x$ = $mate(y)$, otherwise $x$ takes the value NULL. 
After this, the procedure changes the level of $u$ and level of $y$ to 1 and the edge $(u,y)$ is included in the matching. Therefore, Invariant~\ref{invariant 2} is satisfied for vertex $u$. Vertex $u$ and $y$ is removed from the free neighbour list of their neighbours. If $u$ has a free neighbour $w$ then there may be a $3$ length augmenting path involving matched edge $(u,y)$ which violates Invariant~\ref{invariant 5}. Therefore, 
the procedure also checks for any new 3 length augmenting path involving the matched edge $(u,y)$ by using the macro \textit{Check-3-Aug-Path}$(w,u)$. If \textit{Check-3-Aug-Path}$(w,u)$ identifies a $z$ which is not NULL then procedure \textit{Fix-3-Aug-Path-D}$(w,u,y,z)$ is called and  Invariant~\ref{invariant 5} [see Section~\ref{susubsec : Fix-3-Aug-Path-D}] is satisfied by the vertices $(w,u,y,z)$.  Finally, the procedure returns $x$ to the calling procedure.  Also, at the point of return to the calling procedure there are no new invariants violated.
\subsubsection{Deterministic-Raise-Level-To-1$(u)$} 
\label{subsubsec : Deterministic-Raise-Level-To-1}
This procedure receives vertex $u$ as input which is a matched vertex at level 0 and $deg(u) \geq \sqrt{n}$. Therefore, $u$ violates Invariant~\ref{invariant 3}.  
As described in Algorithm~\ref{alg:Deterministic-Raise-Level-To-1}, the procedure works as follows :
Macros \textit{Take-Ownership}$(u)$ and \textit{Transfer-Ownership-To}$(v)$ are executed. The levels of $u$ and $v$ are changed to 1,  thus ensuring that Invariant~\ref{invariant 3} is satisfied for vertex $u$. The control then returns to the calling procedure as there are no new invariants violated by this procedure.
\subsubsection{Randomised-Raise-Level-To-1$(u)$} 
\label{subsubsection : Randomised-Raise-Level-To-1}
This procedure receives a vertex $u$ as input which is a matched vertex at level 0 and $deg(u) \geq \sqrt{n}$. Therefore, $u$ violates Invariant~\ref{invariant 3}. 
As described in Algorithm~\ref{alg:Randomised-Raise-Level-To-1}, procedure \textit{Randomised-Raise-Level-To-1}$(u)$ works as follows :\\ 
Let $v$ be the mate of $u$. The edge $(u,v)$ is removed from the matching. Consequently, Invariant~\ref{invariant 3} is now satisfied for vertex $u$.  \textit{Take-Ownership}$(u)$ is then executed, after which $|O_u| > \sqrt{n}$ and $u$ is free. Now $u$ violates Invariant~\ref{invariant 2} and to repair this,
procedure \textit{Random-Settle-Augmented}$(u)$ is called. This call ensures that 
Invariant~\ref{invariant 2} is satisfied by $u$ [see Section~\ref{subsubsec : Random-Settle-Augmented}]. Further, let $x$ denote the return value. 
if $x$ is not NULL, then 
following two cases describe the remaining processing:
\begin{enumerate}
	\item Level of $x$ is 1. Therefore, $x$ violates Invariant~\ref{invariant_1a}a. Procedure \textit{Handle-Delete-Level1}$(x,1)$ is called. \textit{Handle-Delete-Level1}$(x,1)$ [see Section~\ref{subsubsec : Handle-Delete-Level1}] ensures that $x$ satisfies Invariant~\ref{invariant_1a}a.
	\item Level of $x$ is 0. Therefore, $x$ may violate Invariant~\ref{invariant_1b}. Procedure \textit{Naive-Settle-Augmented}$(x,1)$ is called. \textit{Naive-Settle-Augmented}$(x,1)$ [see Section~\ref{subsubsec : Naive-Settle-Augmented}] ensures that vertex $x$ satisfies Invariant~\ref{invariant_1b}. 
\end{enumerate}
Finally, if $v$ remains free then $v$ may violate Invariant~\ref{invariant_1b}. Procedure \textit{Naive-Settle-Augmented}$(v,1)$ [see Section~\ref{subsubsec : Naive-Settle-Augmented}] is called to ensure that $v$ satisfies Invariant~\ref{invariant_1b}.  The control then returns to the calling procedure as no further invariants are violated due to $u$ selecting a random mate.
\subsubsection{Fix-3-Aug-Path-D$(u,v,y,z)$}
\label{susubsec : Fix-3-Aug-Path-D}
This procedure receives a free vertex $u$ at level $0$, a matched vertex $v \in N(u)$, $y$ = $mate(v)$ and $z \in F(y)$ ($z \neq u$) as input. Further, $u$ and $z$ do not have any free neighbours.  Therefore, there is a 3 length augmenting path $u-v-y-z$ and Invariant~\ref{invariant 5} is violated. As described in Algorithm~\ref{alg:Fix-3-Aug-Path-D}, the procedure works as follows :\\
\textit{Transfer-Ownership-To}$(u)$ and \textit{Transfer-Ownership-To}$(z)$ are executed. Vertices $u$ and $z$ are removed from the free neighbour lists of all their neighbours.  
If level of $v$ is $0$, then \textit{Transfer-Ownership-To}$(v)$ and \textit{Transfer-Ownership-To}$(y)$ are executed, and the levels of $v$ and $y$ are both changed to $1$. Then, the edge $(v,y)$ is removed from the matching and the edges $(u,v)$ and $(y,z)$ are included in the matching. From Observation~\ref{observation : aug-path}, edges $(u,v)$ and $(y,z)$ will not result in a new 3 length augmenting path. Therefore, Invariant~\ref{invariant 5} is satisfied by $u,v,y,z$. Now $u$,$v$,$y$ and $z$ violates Invariant~\ref{invariant 4}. This is addressed by changing the Levels of $u$ and $z$ to 1. Therefore, Invariant~\ref{invariant 4} is satisfied for $u$,$v$,$y$ and $z$. The control then returns to the calling procedure.
\subsubsection{Fix-3-Aug-Path$(u,v,y,z)$} 
\label{subsubsec : Fix-3-Aug-Path}
This procedure receives a free vertex $u$ at level $0$, a matched vertex $v \in N(u)$, $y$ = $mate(v)$ and $z \in F(y)$ ($z \neq u$) as input. Further, $u$ and $z$ do not have a free neighbour.  Therefore, there is a 3 length augmenting path $u-v-y-z$ and Invariant~\ref{invariant 5} is violated.
As described in Algorithm~\ref{alg:Fix-3-Aug-Path}, procedure \textit{Fix-3-Aug-Path}$(u,v)$ proceeds based on the following cases: 
\begin{enumerate} 
	\item Level of vertex $v$ is 1.\\
	The edge $(v,y)$ is removed from the matching, and the edges $(u,v)$ and $(y,z)$ are included in the matching. From Observation~\ref{observation : aug-path}, edges $(u,v)$ and $(y,z)$ will not result in a new 3 length augmenting path. This ensures that Invariant~\ref{invariant 5} is satisfied by $u,v,y,z$. However, $u$,$v$,$y$ and $z$ violates Invariant~\ref{invariant 4}, and this is addressed by considering the following sub-cases:  
	\begin{enumerate}
		\item $deg(u) \geq \sqrt{n}$. Therefore, vertex $u$ also violates Invariant~\ref{invariant 3}. Procedure \textit{Randomised-Raise-Level-To-1$(u)$} is called to address this violation, and on return from the call $u$ satisfies Invariant~\ref{invariant 3} [see Section~\ref{subsubsection : Randomised-Raise-Level-To-1}]. 
		If edge $(u,v)$ is not in the matching on return from the call to procedure \textit{Randomised-Raise-Level-To-1$(u)$}, then $v$ is a free vertex at level $1$. Therefore, $v$ violates Invariant~\ref{invariant_1a}a. Hence, \textit{Handle-Delete-Level1}$(v,1)$ is called and on return $v$ satisfies Invariant~\ref{invariant_1a}a [see Section~\ref{subsubsec : Handle-Delete-Level1}] .
		Following this, \textit{Transfer-Ownership-To}$(z)$ and \textit{Delete-From-F-List}$(z)$ are executed, and then the level of $z$ is changed to $1$. Therefore, $y$ and $z$ satisfy Invariant~\ref{invariant 4}. The control returns to the calling procedure, as no further invariants are violated by this procedure.  
		
		\item $deg(u) < \sqrt{n}$ and $deg(z) \geq \sqrt{n}$. Like $u$ in the previous sub-case, vertex $z$ violates Invariant~\ref{invariant 3}. Procedure \textit{Randomised-Raise-Level-To-1$(z)$} is called to address this violation, and on return from the call $z$ satisfies Invariant~\ref{invariant 3} [see Section~\ref{subsubsection : Randomised-Raise-Level-To-1}].  
		If edge $(y,z)$ is not in the matching on return from the call to \textit{Randomised-Raise-Level-To-1$(z)$}, then $y$ is a free vertex $y$. Therefore, $y$ violates Invariant~\ref{invariant_1a}a. Procedure Handle-Delete-Level1$(y,1)$ is called and on return $y$ satisfies Invariant~\ref{invariant_1a}a [see Section~\ref{subsubsec : Handle-Delete-Level1}].  Then, \textit{Transfer-Ownership-To}$(u)$ and \textit{Delete-From-F-List}$(u)$ are executed, and then Level of $u$ is changed to $1$. This ensures that $u$ and $v$, satisfy Invariant~\ref{invariant 4}.  The control returns to the calling procedure, as no further invariants are violated by this procedure.
		\item $deg(u) < \sqrt{n}$ and $deg(z) < \sqrt{n}$.  \textit{Transfer-Ownership-To}$(u)$ and \textit{Delete-From-F-List}$(u)$ are executed, and this is followed by the execution of
		 \textit{Transfer-Ownership-To}$(z)$ and \textit{Delete-From-F-List}$(z)$. The level of $u$ and $z$ are changed to 1, and this ensure that $u$, $v$, $y$, $z$, satisfy Invariant~\ref{invariant 4}.   The control returns to the calling procedure, as no further invariants are violated by this procedure.  		
	\end{enumerate} 
	\item Level of vertex $v$ is 0.\\
	The edge $(v,y)$ is removed from the matching, and the edges $(u,v)$ and $(y,z)$ are included in the matching. From Observation~\ref{observation : aug-path}, edges $(u,v)$ and $(y,z)$ will not result in a new 3 length augmenting path. This ensures that Invariant~\ref{invariant 5} is satisfied. The following processing is done by considering each of the sub-cases:
	\begin{enumerate}
		\item $deg(u) \geq \sqrt{n}$. In this case $u$ violates Invariant~\ref{invariant 3}. Procedure \textit{Randomised-Raise-Level-To-1}$(u)$ is called and on return $u$ satisfies Invariant~\ref{invariant 3} [see Section~\ref{subsubsection : Randomised-Raise-Level-To-1}]. During the call to \textit{Randomised-Raise-Level-To-1}$(u)$, if $u$ does not get matched to $v$ then $v$ may violate Invariant~\ref{invariant_1b}. Procedure \textit{Naive-Settle-Augmented}$(v,1)$ is called and on return $v$ satisfies Invariant~\ref{invariant_1b} [see Section~\ref{subsubsec : Naive-Settle-Augmented}].
		\item  $deg(z) \geq \sqrt{n}$. In this $z$ violates Invariant~\ref{invariant 3}. Procedure \textit{Randomised-Raise-Level-To-1}$(z)$ is called and on return $z$ satisfies Invariant~\ref{invariant 3} [see Section~\ref{subsubsection : Randomised-Raise-Level-To-1}]. During the call to \textit{Randomised-Raise-Level-To-1}$(z)$, if $z$ does not get matched to $y$ then $y$ may violate Invariant~\ref{invariant_1b}. Procedure \textit{Naive-Settle-Augmented}$(y,1)$ is called and on return $y$ satisfies Invariant~\ref{invariant_1b} [see Section~\ref{subsubsec : Naive-Settle-Augmented}].
		
		\item $deg(u) < \sqrt{n}$. \textit{Delete-From-F-List}$(u)$ is executed to remove $u$ from the free neighbour list of each of its neighbours.
		\item $deg(z) < \sqrt{n}$.  \textit{Delete-From-F-List}$(z)$ is executed to remove $z$ from the free neighbour list of each of its neighbours.
	\end{enumerate}
	The control returns to the calling procedure as no new invariants are violated by this procedure.
\end{enumerate}
\subsubsection{Handle-Delete-Level1$(u,flag)$}
\label{subsubsec : Handle-Delete-Level1}
This procedure receives a free vertex $u$ at level $1$ and a $flag$ as input. Therefore, $u$ violates Invariant~\ref{invariant_1a}a.  
As described in Algorithm~\ref{alg:Handle-Delete-Level1}, procedure \textit{Handle-Delete-Level1}$(u,flag)$ works as follows :\\
\textit{Transfer-Ownership-From}$(u)$ is executed, after which the Level of $u$ is changed to $0$. This ensure that $u$ satisfies Invariant~\ref{invariant_1a}a.
Following two cases describes the remaining processing.
\begin{enumerate}
	\item $|O_u| \geq \sqrt{n}$. In this case $u$ violates Invariant~\ref{invariant 2}. To address this Procedure \textit{Random-Settle-Augmented}$(u)$ is called and on return $u$ satisfies Invariant~\ref{invariant 2} [see Section~\ref{subsubsec : Random-Settle-Augmented}].  Let $x$ be the value returned by \textit{Random-Settle-Augmented}$(u)$. If $x$ is NULL then control returns to the calling procedure as no new invariants are violated by this procedure.  On the other hand, if $x$ is not NULL, then the vertex $x$ may violate Invariant~\ref{invariant_1b}. To address this procedure \textit{Naive-Settle-Augmented}$(x,1)$ is called and on return $x$ satisfies Invariant~\ref{invariant_1b} [see Section~\ref{subsubsec : Naive-Settle-Augmented}]. The control returns to the calling procedure as there are no new violated invariants.
	\item $|O_u| < \sqrt{n}$.  In this case, vertex $u$ may violate Invariant~\ref{invariant_1b}. Therefore, if \textit{flag} is 1 then \textit{Naive-Settle-Augmented}$(u,1)$ is called and if \textit{flag} is 0 then \textit{Naive-Settle-Augmented}$(u,0)$. On return from the call to Procedure \textit{Naive-Settle-Augmented$(u,flag)$} $u$ satisfies Invariant~\ref{invariant_1b} [see Section~\ref{subsubsec : Naive-Settle-Augmented}].  The control returns to the calling procedure as there are no new violated invariants.  
\end{enumerate}
\subsection{Insertion}
\label{subsec : insertion}
We describe procedure calls made during the insertion of an edge $(u,v)$.  
Assuming that all the invariants are maintained before the insertion, we show that our algorithm maintains all the invariants after the insertion.   To do this, we identify the precondition satisfied by the vertices before each of the procedure calls.  
Following four cases exhaustively describes the different insertion cases, and during any insertion exactly one of these cases is executed.\\ 
\textbf{Insert(u,v):}
\begin{enumerate}
	\item
	\label{insertion : case 1} 
	Level of $u$ is 1 and level of $v$ is 1. Edge $(u,v)$ is included in $O_u$ or $O_v$ based on whichever list is larger, and is included in $N(u)$ and  $N(v)$. $u$ and $v$ are clean and no further processing is required, and the update is terminated.
	\item
	\label{insertion : case 2}
	Level of $u$ is $1$ and level of $v$ is $0$. The edge $(u,v)$ is included in $O_u$, $N(u)$ and $N(v)$. Following two sub-cases describes the subsequent steps:
	\begin{enumerate}
		\item
		\label{insertion : case 2a} 
		Vertex $v$ is free.  \textit{Check-3-Aug-Path}$(v,u)$ is executed and it returns a value $z$. If $z$ is not NULL then Invariant~\ref{invariant 5} is violated as there is a 3 length augmenting path starting at $v$. To address this procedure 
		\textit{Fix-3-Aug-Path}$(v,u,mate(u),z)$ is called and on return  Invariant~\ref{invariant 5} [see Section~\ref{subsubsec : Fix-3-Aug-Path}] is satisfied by vertices $(v,u,mate(u),z)$.  
		\item 
		\label{insertion : case 2b}
		Vertex $v$ is matched.  In this sub-case we consider the following sub-cases:
		\begin{enumerate}
			\item $deg(v) = \sqrt{n}$. Then, since $u$ is at level 1, it is clean and $v$ violates Invariant~\ref{invariant 3}.  To address this, Procedure \textit{Randomised-Raise-Level-To-1}$(v)$ is called and on return $v$ satisfies Invariant~\ref{invariant 3} [see Section~\ref{subsubsection : Randomised-Raise-Level-To-1}].
			\item  $deg(v) < \sqrt{n}$. Then, $u$ and $v$ are clean and no further processing is required.
		\end{enumerate}  
		The update is terminated.
	\end{enumerate}
	
	\item Level of $u$ is $0$ and Level of $v$ is $1$. This case is symmetric to case~\ref{insertion : case 2}.
	\item 
	\label{insertion : case 4}
	Level of $u$ is $0$ and level of $v$ is $0$.
	In this case procedure \textit{Handle-Insertion-Level0}$(u,v)$ is called.\\
	As described in Algorithm~\ref{alg:Insertion-Level0}, procedure \textit{Handle-Insertion-Level0}$(u,v)$ works as follows :\\	
	The edge $(u,v)$ is added to $O_u$ if $|O_u | \geq |O_v|$. Otherwise, it is added to $O_v$.
	If both $u$ and $v$ are free vertices then $u$ and $v$ violate Invariant~\ref{invariant_1b}. 
	Then, edge $(u,v)$ is added to the matching and thus $u$ and $v$ satisfy Invariant~\ref{invariant_1b}.  Without loss of generality, let us assume that $|O_u| \geq |O_v|$. The following cases describe the remaining steps:
	\begin{enumerate}
		\item $|O_u| = \sqrt{n}$. In this case $u$ violates Invariant~\ref{invariant 2}. To address this procedure \textit{Random-Settle-Augmented}$(u)$ is called and on return $u$ satisfies Invariant~\ref{invariant 2} [see Section~\ref{subsubsec : Random-Settle-Augmented}]. Let $x$ be the value returned by \textit{Random-Settle-Augmented}$(u)$. If $x$ is not NULL then $x$ may violate Invariant~\ref{invariant_1b} and this is addressed by calling Procedure \textit{Naive-Settle-Augmented}$(x,1)$  and on return $x$ satisfies Invariant~\ref{invariant_1b} [see Section~\ref{subsubsec : Naive-Settle-Augmented}]. Now we deal with the $v$, the other end point of the inserted edge $(u,v)$ in the following two sub-cases. 
	\begin{enumerate}
	\item  Suppose $v$ was matched to $u$ at the beginning, and after the call to \textit{Random-Settle-Augmented}$(u)$, $v$ is free. Then, $v$ may violate Invariant~\ref{invariant_1b}. This is addressed by a call to procedure  \textit{Naive-Settle-Augmented}$(v,1)$ and on return $v$ satisfies Invariant~\ref{invariant_1b} [see Section~\ref{subsubsec : Naive-Settle-Augmented}]. 
	\item $v$ was not matched to $u$ at the beginning, but at this point $v$ is matched, $deg(v) \geq \sqrt{n}$ and level of $v$ is $0$. Then $v$ violates Invariant~\ref{invariant 3}. This is addressed by a call to procedure \textit{Deterministic-Raise-Level-To-1}$(v)$  and on return $v$ satisfies Invariant~\ref{invariant 3} [see Section~\ref{subsubsec : Deterministic-Raise-Level-To-1}].
	\end{enumerate}	
	\item $|O_u| < \sqrt{n}$. We deal with this case by executing exactly one of the following two sub-cases:
		\begin{enumerate}
			\item vertex $v$ is matched.  The steps in this case are chosen based on exactly one of the following three sub-cases:
			\begin{enumerate}
				\item $deg(v) \geq \sqrt{n}$. Then, $v$ violates Invariant~\ref{invariant 3}.  This is addressed by a call to procedure \textit{Random-Raise-Level-To-1}$(v)$  and on return $v$ satisfies Invariant~\ref{invariant 3} [Section~\ref{subsubsection : Randomised-Raise-Level-To-1}]. If $u$ is matched, $deg(u) \geq \sqrt{n}$ and level of $u$ is $0$ then, $u$ violates Invariant~\ref{invariant 3}. To address this procedure \textit{Deterministic-Raise-Level-To-1}$(u)$ is called and on return $u$ satisfies Invariant~\ref{invariant 3} [see Section~\ref{subsubsec : Deterministic-Raise-Level-To-1}].
				\item $deg(v) < \sqrt{n}$ and $u$ is free.
				\textit{Check-3-Aug-Path}$(u,v)$ is executed and let $z$ be the value computed by it. If $z$ is not NULL then Invariant~\ref{invariant 5} is violated, and 
				procedure \textit{Fix-3-Aug-Path}$(u,v,mate(v),z)$ is called. On return from the call Invariant~\ref{invariant 5} is  satisfied  [see Section~\ref{subsubsec : Fix-3-Aug-Path}] by vertices $(u,v,mate(v),z)$.
				\item $deg(v) < \sqrt{n}$, $u$ is matched and $deg(u) \geq \sqrt{n}$. Then, $u$ violates Invariant~\ref{invariant 3}. This is addressed by a call to procedure \textit{Random-Raise-Level-To-1}$(u)$ and on return $u$ satisfies Invariant~\ref{invariant 3} [see Section~\ref{subsubsection : Randomised-Raise-Level-To-1}]. 
			\end{enumerate}
			\item vertex $v$ is free. This case is handled by considering exactly one of the following two sub-cases:
			\begin{enumerate}
				\item $u$ is matched and $deg(u) \geq \sqrt{n}$. $u$ violates Invariant~\ref{invariant 3}. This is addressed by a call to procedure \textit{Random-Raise-Level-To-1}$(u)$ and on return $u$ satisfies Invariant~\ref{invariant 3} [see Section~\ref{subsubsection : Randomised-Raise-Level-To-1}].
				\item $u$ is matched, $deg(u) < \sqrt{n}$ and $v$ is free. 
				\textit{Check-3-Aug-Path}$(v,u)$ is executed and let $z$ be the value computed by it.  If $z$ is not NULL then Invariant~\ref{invariant 5} is violated, and procedure \textit{Fix-3-Aug-Path}$(v,u,mate(u),z)$ is called. On return from the call Invariant~\ref{invariant 5} is satisfied [see Section~\ref{subsubsec : Fix-3-Aug-Path}] by vertices $(v,u,mate(u),z)$. 
			\end{enumerate}
		\end{enumerate}   
	\end{enumerate}
	Return to calling procedure as no new invariants are violated and the update is terminated.
	
\end{enumerate}
\subsection{Deletion}
\label{subsec : deletion}
We describe the procedure calls made during the deletion of an edge $(u,v)$.   
Assuming that all the invariants are maintained before the deletion, our aim is to ensure that our algorithm maintains all the invariants after the deletion.  Deletion executes exactly one of the following two cases :\\
\textbf{Delete(u,v):}
\begin{enumerate}
	\item $(u,v)$ is an unmatched edge. This case is simple and does not violate any invariant. Our algorithm removes $(u,v)$ from $O_u$ or $O_v$ and from $N(u)$ and $N(v)$. This takes $O(1)$ time. This deletion does not change the matching and the update is terminated. The processing time is $O(1)$. 
	\item $(u,v)$ is a matched edge. One of the following two sub-cases based on the level of $(u,v)$ is executed.
	\begin{enumerate}
		\item 
		\label{deletion : case 2a}
		$(u,v)$ is a level 0 matched edge. Therefore, $u$ and $v$ may violate Invariant~\ref{invariant_1b}. To address this procedures \textit{Naive-Settle-Augmented}$(u,0)$ and \textit{Naive-Settle-Augmented}$(v,0)$ are called. On return from the calls both $u$ and $v$ satisfy Invariant~\ref{invariant_1b} [see Section~\ref{subsubsec : Naive-Settle-Augmented}].   This terminates the update.
		
		\item 
		\label{case : deletion 2b}
		$(u,v)$ is a level 1 matched edge. Deletion of $(u,v)$ creates two free vertices $u$ and $v$ at level 1. Therefore, $u$ and $v$ violate Invariant~\ref{invariant_1a}a. For vertex $u$, procedure \textit{Handle-Delete-Level1}$(u,0)$ is called. Similarly for $v$, procedure \textit{Handle-Delete-Level1}$(v,0)$ is called. 
		After the two calls, \textit{Handle-Delete-Level1}$(u,0)$ and \textit{Handle-Delete-Level1}$(v,0)$, both $u$ and $v$ satisfy Invariant~\ref{invariant_1a}a [see Section~\ref{subsubsec : Handle-Delete-Level1}].  This terminates the update.
	\end{enumerate}
\end{enumerate}
Before presenting the correctness, we make the following observation :
\begin{observation}
\label{obs : flag-change-0-to-1}	
The value of flag changes from $0$ to $1$ only after a call to procedure  \textit{Random-Settle-Augmented} within the same update.
\end{observation}

\section{Correctness-Termination of Updates and Maintenance of Invariants}
\label{sec: Algorithm_correctness}
To prove that the procedures described in the previous section are correct, we prove that each update terminates and at the end of each update all the vertices are clean.  We prove a stronger statement that each update terminates by making at most a constant number of procedure calls.  
\begin{lemma}
\label{lem: ownership_edge}	
After an update each edge is owned by exactly one of its two vertices. Further, if the vertices of an edge are at different levels, then the edge is owned by the vertex at level 1.
\end{lemma}
\begin{proof}
Before the first update, the graph is empty.  Therefore, this claim is true at the before of the first update.  Consequently, we assume that before an update the claim is true. We prove that after the update, the claim is true.   From the description of the insert procedure, it is clear that when an edge is inserted, it is always added into the ownership list of exactly one of the two vertices it is incident on.  Similarly, edge deletion does not violate the property.  We now show based on the following two cases that after a vertex changes its levels, the condition in the lemma is respected by all the edges incident on it. 
\begin{enumerate}
\item A vertex $u$ changes it level from $0$ to $1$ :
 This happens inside  \textit{Deterministic-Raise-Level-To-1}, \textit{Random-Settle-Augmented}, \textit{Fix-3-Aug-Path} and \textit{Fix-3-Aug-Path-D}.  In these procedures, before changing the level of $u$ from $0$ to $1$, we execute \textit{Take-Ownership}$(u)$ or  \textit{Transfer-Ownership-To}$(u)$.  \textit{Take-Ownership}$(u)$ transfers the ownership of all the edges incident on $u$, but not owned by $u$, to $u$.  \textit{Transfer-Ownership-To}$(u)$ transfers the ownership of all the edges whose other endpoint is at level $0$, and not owned by $u$, to $u$. This ensures that all the edges incident on $u$ satisfies the statement of this lemma.
\item A vertex $u$ changes it level from $1$ to $0$ : 
This happens only within the procedure \textit{Handle-Delete-Level1}.  Before changing the level of $u$ from $1$ to $0$, we execute \textit{Transfer-Ownership-From}$(u)$.  \textit{Transfer-Ownership-From}$(u)$  considers each edge $(u,x)$ in $O_u$ such that $x$ is at level 1, and removes it from $O_u$ and adds it to $O_x$.   This ensures that all the edges incident on $u$ satisfies the statement of this lemma.
\end{enumerate}   
Therefore, at the end of an update, all the edges satisfy the statement of the lemma, thus the lemma is proved.
\end{proof} 
We next prove that when a 3 length augmenting path is removed by \textit{Fix-3-Aug-Path} or \textit{Fix-3-Aug-Path-D}, they do not create a new 3 length augmenting path with respect to the modified matching.  We prove this based on the following crucial property.
The lemma points out that our update procedures implement an {\em atomic operation} to match at least one of the two end points of an edge when they both become free during an update.  
\begin{lemma}
\label{lem:noadj}
Let us assume that all the vertices are clean at the beginning of an update.  For an edge $(u,p)$, if $u$ and $p$ become free after a statement in some procedure during the update, then at least one of $u$ and $p$ is   matched before within the next two statements in the procedure.
\end{lemma}
\begin{proof}
Let us assume that during an update two vertices $u$ and $p$ are free and $(u,p)$ is an edge. Let us consider the statement  just after which  both $u$ and $p$ are free. We now consider two cases:\\
Without loss of generality let us consider the case that $p$ became free after $u$.   Consequently, just before the statement in which $p$ becomes free the vertex $u$ is present in $F(p)$. The reason for this is as follows.  Consider the last statement after which $u$ became free before the current statement after which $p$ has become free.  We know from the design of the procedures that there would have been an attempt to find a mate for $u$ (using \textit{Naive-Settle-Augmented} or \textit{Random-Settle-Augmented}).  However, since $u$ failed to be matched it would have been inserted into the free neighbour list of all its neighbours.  Therefore, $u$ is in $F(p)$ just before the statement after which $p$ has become free.  Once $p$ has become free, from the description of the procedures in Section \ref{algorithm} we know that it is immediately addressed by the procedures. These procedures immediately try to find a mate for $p$ and hence they would have matched $p$ to some free neighbour, which definitely exists in this case, since $u$ is in $F(p)$.  Therefore, in this case the lemma is proved. \\
Secondly, let us consider the case when there is  a single statement after which $p$ and $u$ are free. In this case, it means that the edge $(u,p)$ was in the matching and it was removed from the matching.  Such a statement occurs in procedures \textit{Random-Raise-Level-To-1}, \textit{Random-Settle-Augmented}, \textit{Fix-3-Aug-Path} and \textit{Fix-3-Aug-Path-D}. In the cases of  \textit{Fix-3-Aug-Path} and \textit{Fix-3-Aug-Path-D} it is ensured that both the vertices are immediately matched after $(u,p)$ is removed from the matching.  In the case of \textit{Random-Settle-Augmented} such a removal of an edge form the matching is followed by matching one end point, say $u$,  to a different vertex, and the \textit{Fix-3-Aug-Path-D} called inside \textit{Random-Settle-Augmented} will ensure that $u$ continues to be matched.  Finally, in \textit{Random-Raise-Level-To-1}, one of the two free vertices, say $u$, is forcibly matched in a call to \textit{Random-Settle-Augmented}, and on return at least one of $u$ and $p$ is matched.  Hence the lemma.
\end{proof}
As a consequence of Lemma \ref{lem:noadj} we prove the following lemma to ensure that if we remove a 3 length augmenting path, then the newly matched vertices are not in a 3 length augmenting path with respect to the new matching.
\begin{lemma}
\label{lem:safez}
Let $u$ be a free vertex such that each neighbour is matched, let $v$ be a neighbour of $u$, and let $y = mate(v)$.  During an update if the macro \textit{Check-3-Aug-Path$(u, v)$ } finds a non-NULL $z$ which is a free neighbour of $y$, then each neighbour of $z$ is matched in the matching maintained during the update.  Therefore, after \textit{Fix-3-Aug-Path$(u,v,y,z)$} and \textit{Fix-3-Aug-Path-D$(u,v,y,z)$} there is no
3 length augmenting path involving the vertices $u$ and $z$.  
\end{lemma}
\begin{proof}
Since $z$ is free after the execution of the macro \textit{Check-3-Aug-Path$(u, v)$}, $z$ would have been free prior to the execution of  \textit{Check-3-Aug-Path$(u, v)$}.   The reason for this is that \textit{Check-3-Aug-Path$(u,v)$} does not change the matching and only looks for a free vertex.  Therefore, prior to the execution of \textit{Check-3-Aug-Path$(u,v)$} $z$ is free.  We now prove that just  before the execution of \textit{Check-3-Aug-Path$(u, v)$}, $z$ does not have a free neighbour. On the other hand, let us assume that  $z$ is free and has a free neighbour $p$ just before execution of \textit{Check-3-Aug-Path$(u, v)$}.   From Lemma \ref{lem:noadj}, it follows that within the following two statements either $z$ or $p$ would have been matched.   However, we know that none of the statements in \textit{Check-3-Aug-Path$(u, v)$} changes the matching.  Therefore, our assumption that $z$ has a free neighbour just before the execution of \textit{Check-3-Aug-Path$(u, v)$} is wrong.  Consequently, just before the execution of \textit{Check-3-Aug-Path$(u, v)$} $z$ is free and all its neighbours would be matched vertices.  This proves the first statement.

Secondly, since $u$ and $z$ have no free neighbours, from Observation \ref{observation : aug-path} it follows that after 
\textit{Fix-3-Aug-Path$(u,v,y,z)$} or \textit{Fix-3-Aug-Path-D$(u,v,y,z)$} there is no
3 length augmenting path involving the vertices $u$ and $z$.  Hence the lemma.
\end{proof}
\begin{lemma}
\label{thm:number-of-calls-perupdate}
Each update terminates after making a constant number of procedure calls.  The constant is at most 30.
\end{lemma}
\begin{proof}
Our proof approach is to show that  for each procedure call, say $P$, from a calling procedure $Q$, the control returns to $Q$ after making at most a constant number of procedure calls.  Assuming this is true, we show that the update functions {\em insert} and {\em delete} terminate by making at most a constant number of procedure calls.  This claim is immediately seen to be true, because during run-time the updates select exactly one case to execute.  Each case in an update is a constant length sequence of procedure calls.  Thus if each of the procedure call returns to the calling function after making at most a constant number of procedure calls, it follows that the updates terminate after making at most a constant number of procedure calls.  Further, a crucial observation is that the statements other than the procedure calls in each procedure all terminate as the loops are all of finite length, and consist only of assignment statements.  We now show that each procedure call returns after making at most a constant number of procedure calls.   Our proof is based on the description of the procedures, Table~\ref{table : Procedure_Description}, and the function call graph described in   Table~\ref{table : Procedure_Calls}.  It is immediately true that each case in each procedure has at most a constant number of function calls. However, the run-time analysis crucially depends on the value of $flag$.  We now do case-wise analysis of the number of procedure calls  made at run-time in each procedure.  We present the procedures in non-decreasing order of the number of procedure calls made.
\begin{enumerate}
\item Procedures \textit{Deterministic-Raise-Level-To-1} and \textit{Fix-3-Aug-Path-D} do not call any other procedure. Therefore, they return to the calling procedure on completion.  
\item The only procedure call made by \textit{Random-Settle-Augmented} is to procedure \textit{Fix-3-Aug-Path-D}. Therefore, a call to  \textit{Random-Settle-Augmented} returns to the calling procedure after making at most one more call.
\item Procedure \textit{Naive-Settle-Augmented} with $flag$ value $1$ either returns without making another procedure call
 or it makes a call to either \textit{Deterministic-Raise-Level-To-1} or \textit{Fix-3-Aug-Path-D}. Therefore, a call to \textit{Naive-Settle-Augmented} with $flag$ value $1$ returns to the calling procedure after making at most one more call.   
\item Procedure \textit{Handle-Delete-Level1} with $flag$ value $1$ makes at most one call to procedure \textit{Random-Settle-Augmented} and at most one call to \textit{Naive-Settle-Augmented} with $flag$ value $1$. Therefore, a call to \textit{Handle-Delete-Level1} with $flag$ value $1$ returns to the calling procedure after making at most four more calls.   
\item Procedure \textit{Random-Raise-Level-To-1} makes a call to \textit{Random-Settle-Augmented}. From item $2$ above, a call to \textit{Random-Settle-Augmented} returns to the calling procedure after making at most one more call. After this \textit{Random-Raise-Level-To-1} makes a call to either \textit{Handle-Delete-Level1} with $flag$ value $1$ or \textit{Naive-Settle-Augmented} with $flag$ value $1$. From item $3$ above, a call to \textit{Naive-Settle-Augmented} with $flag$ value $1$ terminates with at most one more call. From item $4$ above, a call to procedure \textit{Handle-Delete-Level1} with $flag$ value $1$ returns to the calling procedure after making at most four more calls. Finally, \textit{Random-Raise-Level-To-1} makes at most one more call to \textit{Naive-Settle-Augmented} with $flag$ value $1$. Therefore, a call to \textit{Random-Raise-Level-To-1} returns to the calling procedure after making at most eight more calls.
\item Procedure \textit{Fix-3-Aug-Path} either returns to the calling procedure without making any further calls or makes a call to procedure \textit{Random-Raise-Level-To-1}. From item $5$ above, a call to \textit{Random-Raise-Level-To-1} returns to the calling procedure after making at most eight more calls.  Following this it either makes  a call to either \textit{Handle-Delete-Level1} with $flag$ value $1$ or \textit{Naive-Settle-Augmented} with $flag$ value $1$. From items $4$ and $3$ above, these two calls return after making four more and one more procedure call, respectively.  Therefore, \textit{Fix-3-Aug-Path} returns to the calling procedure after atmost 12 procedure calls.
\item Procedure \textit{Naive-Settle-Augmented} with $flag$ value $0$ either returns to the calling procedure without making any further procedure calls or does one of the following :
\begin{enumerate}
\item Makes a call to procedure \textit{Random-Raise-Level-To-1}. From item $5$ above, a call to \textit{Random-Raise-Level-To-1} returns after making at most eight more calls. After this, \textit{Naive-Settle-Augmented} with $flag$ value $0$ makes at most one more call to \textit{Naive-Settle-Augmented} with $flag$ value $1$. From item $3$, a call to \textit{Naive-Settle-Augmented} with $flag$ value $1$ returns after making at most one more call.
\item Makes a call to procedure \textit{Fix-3-Aug-Path}. From item $6$ above, a call to procedure \textit{Fix-3-Aug-Path} returns after making at most twelve more calls.
\end{enumerate}
Therefore, a call to \textit{Naive-Settle-Augmented} with $flag$ value $0$ returns to the calling procedure after making at most thirteen more calls.    
\item Procedure \textit{Handle-Delete-Level1} with $flag$ value $0$ does one of the following :
\begin{enumerate}
\item Makes a call to \textit{Random-Settle-Augmented}. From item $2$ above, a call to \textit{Random-Settle-Augmented} returns to the after making at most one more call. After this \textit{Handle-Delete-Level1} with $flag$ value $0$ makes at most one more call to \textit{Naive-Settle-Augmented} with $flag$ value $1$. From item $3$ above, a call to \textit{Naive-Settle-Augmented} with $flag$ value $1$ which returns after making at most one more call. 
\item Makes a call to \textit{Naive-Settle-Augmented} with $flag$ value $0$. From item $7$, a call to \textit{Naive-Settle-Augmented} with $flag$ value $0$ returns after making at most thirteen more calls. 	
\end{enumerate}
Therefore, a call to \textit{Handle-Delete-Level1} with $flag$ value $0$ returns to the calling procedure after making at most fourteen more calls.
\end{enumerate} 
From the above analysis, procedure \textit{Handle-Delete-Level1} with flag value 0 makes the maximum number of procedure calls, which is 14 calls, before returning to the calling procedure.   The maximum number of procedure calls made by an insert is at most 13, and the maximum number of procedures called by a delete is at most 30.  Therefore, each update terminates by making at most 30 procedure calls.  Hence the lemma.
\end{proof} 
\begin{theorem}
\label{thm : clean_update}	
After the termination of each update all the vertices are clean.   Consequently, after each update a maximal matching without 3 length augmenting paths is maintained.
\end{theorem}
\begin{proof}
We prove this by induction on the number of updates.  When the number of updates is 0, the graph is empty and all the vertices are clean. Let the claim be true after $i > 0$ updates. In other words words, after $i$ updates, all the vertices are clean.  We now prove that after the $i+1$-th update, whether it an Insert or delete, all the vertices are clean.  Our proof is by observing properties of the control flow in our procedures.  The violated invariants are propositions involving the size of the ownership list, vertex degree, matched or unmatched state of a vertex, level of a vertex, and the presence of a 3 length augmenting path.   Any violated invariant at a vertex is fixed by a combination of the operations involving the modification of its ownership list, changing its level,  finding a mate, and exchanging the matching and non-matching edges in a 3 length augmenting path. 
In each of these procedures in which these operations are performed, just after the operation, we check for the violation of an invariant  and invoke appropriate procedures to fix the violated invariants.   Further, since all vertices are clean at the beginning of the update and from Lemma \ref{lem:safez}, it follows that whenever a 3 length augmenting path is fixed, no new 3 length augmenting paths with respect to the new matching are created.   This is the reason why  in our procedures we do not check if fixing a 3 length augmenting path creates new 3 length augmenting paths.   Therefore, when the control exits from the update function, all vertices are clean.  From Lemma~\ref{thm:number-of-calls-perupdate} it follows that each update step terminates, and therefore at the end of $i+1$-th update all the vertices are clean.   
Since, all the vertices are clean, it follows that the neighbours of all free vertices are matched and there are no 3 length augmenting paths.  Therefore, the matching at the end of each update is a maximal matching without 3 length augmenting paths.  Hence the Theorem is proved.
\end{proof}

\section{Analysis of the Expected Amortized Update Time}
\label{epochanalysis}
In this section we present our upper bound on the expected value of the total update time of our algorithm on an update sequence.  
For the asymptotic analysis of the expected total update time, we consider an extended update sequence which has the additional property that at the end the graph is empty.   The extended update sequence we consider for the analysis is obtained from the given update sequence by performing a sequence the delete updates at the end till all the edges are deleted.   Note that if the original update sequence had $t$ updates, then the extended update sequence with the additional deletes has at most $t'=2t$ updates.   Further, note that the update sequence starts from the empty graph. This is crucial in the proof of Theorem \ref{thm : clean_update}	where before any update all the vertices are clean.
\begin{observation}
\label{extendupdate}
The expected total update time for the given update sequence is at most the expected total update time for the extended sequence.  
\end{observation}
In our analysis we crucially use that at the end of the update sequence the graph is empty.  
Our analysis is by extending the concept of \textit{epochs} (Definition~\ref{def:epoch}) from \cite{DBLP:BGS}. Our  approach is summarized in the following sequence of analysis steps:
\begin{itemize}
\item We show that total update time is given by the sum total of the creation time and termination time of the epochs associated with each procedure call.  
\item In Section~\ref{subsec: Procedures and associated epochs} we associate the computation time of different procedure calls with creation time and termination time of different epochs. 
\item To bound the total update time, we classify the epochs into level $0$ and level $1$. Properties of these epochs are presented in Section~\ref{subsec: Procedures and associated epochs}
.  Similar to \cite{DBLP:BGS} we use the fact that each level 0 epoch has a worst-case time of $O(\sqrt{n})$ associated with it.
\item We  then classify the level $1$ epochs into two types which we introduce, based on how they are created: Random Level 1 Epochs and Deterministic Level 1 Epochs.  The Deterministic level 1 epochs are further classified into Type 1 and Type 2 {\em inexpensive } epochs based on their contribution to total time in Section~\ref{subsubsec: Random-Deterministic-Level1-epochs}.  Here,  we bound the contribution of Type 1  Inexpensive Deterministic epochs to the total update time.  This  has an amortized cost of $O(\sqrt{n})$.
\item Finally, we define \textit{epoch-sets} (Definition~\ref{def:epoch-set}). An  epoch-set consists of one random level 1 epoch and a constant number of Deterministic epochs. These Deterministic epochs are not type 1 inexpensive, and are either type 2 inexpensive epochs or those which are not inexpensive epochs.   We then upper bound the expected value of total update time using linearity of expectation over the epoch-sets in Section~\ref{subsec : expected-total-update-time}. We then analyse the worst case total update time of our algorithm with high probability in Section~\ref{subsec:worstcase-total-updatetime}. 
\end{itemize}
\begin{definition}
\label{def:epoch}	
\cite{DBLP:BGS}	
At a time instant $t$, let $(u,v)$ be an edge in the matching $M$. The epoch defined by edge $(u,v)$ at time $t$ is the maximal continuous time interval such that the interval contains $t$ and during the interval $(u,v)$ is  in $M$. An epoch is said to be a level $0$ epoch or level $1$ epoch depending upon the level of the matched edge that defines the epoch.	
We refer to the an epoch by the matching edge associated with the epoch. For example, when we refer to the epoch $(u,v)$ we mean an epoch associated with the matching edge $(u,v)$.  We also refer to an epoch at a vertex $u$. This refers to an epoch $(u,v)$.  
\end{definition}
{\bf Total Update Time via  time associated with the creation and termination of Epochs:}
If we fix an edge $(u,v)$ and consider the time period from the first insertion of edge $(u,v)$ in the graph till its final deletion from the graph by the extended update sequence, then this period consists of a sequence of epochs separated by the maximal continuous time periods during which $(u,v)$ is not in the matching.  Note that this sequence could even be empty, and this happens if $(u,v)$ is not in any matching maintained by the algorithm throughout all the updates.  From the description of insertion in section~\ref{subsec : insertion} and deletion in section~\ref{subsec : deletion}, it is clear that any update operation that does not change the matching is processed in $O(1)$ time.
Further, if an update changes the matching, then the change is done by a sequence of procedure calls.  Each such procedure call changes  the matching by adding or deleting edges from the matching.  Consequently, each procedure call is associated with the creation of some new epochs and the termination of some existing epochs.    We  associate the total computation performed for every update operation with the creation and termination of different epochs which takes place inside different procedure calls during processing of the update.  From Lemma \ref{thm:number-of-calls-perupdate} we know that each update terminates after a constant number of procedure calls, and thus each update creates and terminates at most a constant number of epochs.   We formally use this observation after the necessary set-up to bound the expected value of the total update time in Section ~\ref{subsec : expected-total-update-time}.
\subsection{Epochs associated with each procedure }
\label{subsec: Procedures and associated epochs}
In this section we associate a set of epochs with every procedure.  The time spent in the procedure is suitably distributed to the creation and termination of the epochs.  
This association is based on the description of the procedures in Section~\ref{subsec : Description_Procedures} and the description of \textit{Handle-Insert-Level0} in Section~\ref{subsec : insertion}. in Table \ref{table : Procedure_Calls_VS_Epocs} we present the epochs created and terminated by the procedures and \textit{Handle-Insert-Level0}.
The key property that we ensure is that the computation time of a  procedure is associated to the  creation and termination of suitably identified level 0 and level 1 epochs.
 In most cases, the addition or removal of a level 0 or level 1 matched edge from the matching in the body of a procedure (meaning, not in the procedure calls made inside it), corresponds to the epochs created or terminated by the procedure. 
 The time associated with the creation and termination of the epochs is the time spent, in the procedure or just before entry into the procedure, towards the addition or removal of matching edges.  
\begin{table}[h!]
	\centering
	\begin{tabular}{|m{3cm}|m{5cm}|m{5cm}|m{2cm}|}
		\hline
		\textbf{Procedure} & \textbf{Condition} & \centering \textbf{Associated Epoch} & \textbf{Computation Time}\\
		\hline
		\multirow{3}{3cm}{Naive-Settle-Augmented(u,flag)}
		& $F(u)$ is empty and there is no $3$ length augmenting path starting at $u$. A call to this has happened from Handle-Delete-Level1$(u,flag)$ & Termination of level $1$ epoch $(u,u^{\prime})$ where $u^{\prime}$ is previous mate of $u$ at level $1$  & $O(deg(u))$\\
		\cline{2-4}
		& $F(u)$ is empty and there is no $3$ length augmenting path starting at $u$ & Termination of level $0$ epoch $(u,u^{\prime})$ where $u^{\prime}$ is previous mate of $u$ at level $0$ & $O(\sqrt{n})$\\
		\cline{2-4}
		& $w \in F(u)$ and $deg(u) < \sqrt{n}$ and $deg(w) < \sqrt{n}$ & Creation of level $0$ epoch $(u,w)$ & $O(\sqrt{n})$ \\
		\hline
		Random-Settle-Augmented(u)
		& An edge is selected uniformly at random from $O_u$, say $(u,y)$. Edge $(u,y)$ is included in the matching and Level of $u$ and $y$ are $1$ & Creation of level $1$ epoch $(u,y)$ and if $y$ was matched then termination of epoch $(y,mate(y))$ & $O(deg(u) 
		+deg(y))$ \\
		\hline
		Deterministic-Raise-Level-To-1(u)
		& Levels of $u$ and mate of $u$, say $u^{\prime}$ are changed to $1$ & Creation of level $1$ epoch $(u,u^{\prime})$ where $u^{\prime}$ is mate of $u$ & $O(deg(u) 
		+deg(u^{\prime}))$
		\\
		\hline
		Randomised-Raise-Level-To-1(u)
		& Level of $u$ is $0$, $u$ is matched and $deg(u) \geq \sqrt{n}$. This subsequently calls Random-Settle-Augmented(u). &
		Termination of epoch $(u,mate(u))$ and creation of level 1 epoch $(u,y)$  & $O(deg(u))$ 
		\\
		\hline
		Fix-3-Aug-Path-D$(u,v,y,z)$ & $u$ is in $F(v)$, $v$ is matched and Level of $v$ is 1 or 0. mate of $v$ is $y$ and $z$ is in $F(y)$ & Termination of epoch $(v,y)$ and creation of level $1$ epochs $(u,v)$ and $(y,z)$ & $O(deg(u)+deg(z))$
		\\
		\hline
		\multirow{7}{3cm}{Fix-3-Aug-Path(u,v,y,z)}
		& \emph{common condition : mate of $v$ is $y$, $z$ is in $F(y)$} & &\\
		\cline{2-4}
		& Level of $v$ is $1$, 
		$deg(u) \geq \sqrt{n}$
		& Termination of level 1 epoch $(v,y)$ and creation of level $1$ epoch $(y,z)$
		& $O(deg(z))$ \\
		\cline{2-4}
		& Level of $v$ is $1$, 
		$deg(u) < \sqrt{n}$ and $deg(z) \geq \sqrt{n}$
		&
		Termination of level 1 epoch $(v,y)$ and
	    creation of level $1$ epoch $(u,v)$
		 & $O(deg(u))$ \\
		 \cline{2-4}
		& Level of $v$ is $1$, $deg(u) < \sqrt{n}$ and $deg(z) < \sqrt{n}$ & Termination of level 1 epoch $(v,y)$ and creation of level $1$ epochs $(u,v)$ and $(y,z)$ & $O(\sqrt{n})$
		\\
		\cline{2-4}
		& Level of $v$ is $0$, $deg(u) \geq \sqrt{n}$ and $deg(z) < \sqrt{n}$ & Termination of level 0 epoch $(v,y)$ and creation of level $0$ epochs $(y,z)$ & $O(\sqrt{n})$
		\\
		\cline{2-4}
		& Level of $v$ is $0$, $deg(u) < \sqrt{n}$ and $deg(z) \geq \sqrt{n}$ & Termination of level 0 epoch $(v,y)$ and creation of level $0$ epochs $(u,v)$ & $O(\sqrt{n})$
		\\
		\cline{2-4}
		& Level of $v$ is $0$, $deg(u) < \sqrt{n}$ and $deg(z) < \sqrt{n}$ & Termination of level 0 epoch $(v,y)$ and creation of level $0$ epochs $(u,v)$ and $(y,z)$ & $O(\sqrt{n})$
		\\
		\hline
		Handle-Delete-Level1(u,flag)
		& Level of $u$ changes to $0$ and $u$ remains as a free vertex & Termination of level $1$ epoch $(u,u^{\prime})$ where $u^{\prime}$ is previous mate of $u$ & $O(deg(u))$
		\\
		\hline
		\multirow{2}{3cm}{Handle-Insert-Level0(u,v)}
		& Levels of $u$ and $v$ are $0$, $u$ and $v$ are free, and $|O_u| < \sqrt{n}$,$|O_v| < \sqrt{n}$, $deg(u) < \sqrt{n}$ and $deg(v) < \sqrt{n}$  & Creation of level $0$ epoch $(u,v)$  & $O(\sqrt{n})$\\
		\hline
\end{tabular}
\caption{Association of  procedures in Section \ref{algorithm} and \textit{Handle-Insert-Level0} with the corresponding Epochs and their creation or termination time}
\label{table : Procedure_Calls_VS_Epocs}
\end{table}
In the following lemmas we present the computation time associated with the creation and termination of level $0$ and level $1$ epochs.  The main aim of the the lemmas is to present a detailed description of the fact that all the computation time during the updates is associated the creation and termination of different epochs. 
\begin{lemma}
\label{lem: Computation-Time-Level0-epochs}
The computation time associated with the creation of a level $0$ epoch is $O(\sqrt{n})$ and termination of a level $0$ epoch is $O(\sqrt{n})$.
\end{lemma}
\begin{proof}
We perform an exhaustive case analysis for the creation and termination of level $0$ epochs by identifying the specific operations that account for the time associated with the epochs.\\	
\textbf{Creation of Level $0$ epochs}:\\
An epoch in level $0$ is created in the body of the following procedure calls :
\begin{enumerate}
\item \textit{Naive-Settle-Augmented$(u,flag)$} : An  epoch is created by this procedure only if $u$ has a free neighbour $w$, and this is found and included in the matching in $O(\sqrt{n})$ time.  The only case where no other procedure call is made is when
$deg(u) < \sqrt{n}$ and $deg(w) < \sqrt{n}$, and in this case  $u$ and $w$ are removed from the free neighbour lists of all their neighbours. Total computation involved is $O(deg(u)+deg(w))$ = $O(\sqrt{n})$. This $O(\sqrt{n})$ computation time is associated with  creation of the level $0$ epoch $(u,w)$.
\item \textit{Fix-3-Aug-Path$(u,v,y,z)$}: At most two epochs are created  by this procedure where $u$ is a free vertex at level $0$, $v$ is a matched vertex at level $0$, $y$ is mate of $v$ and $z$($z \neq u$) is a free neighbour of $y$. Edge $(v,y)$ is removed from the matching and edges $(u,v)$ and $(y,z)$ are added to the matching in $O(1)$ time. If $deg(u) < \sqrt{n}$ then vertex $u$ is removed from free neighbour list of all its neighbours in $O(deg(u))$ = $O(\sqrt{n})$ time. This $O(\sqrt{n})$ computation is associated with creation of level $0$ epoch $(u,v)$. If $deg(z) < \sqrt{n}$ then vertex $z$ is removed from free neighbour list of all its neighbours in $O(deg(z))$ = $O(\sqrt{n})$ time. This $O(\sqrt{n})$ computation time is associated with  creation of the level $0$ epoch $(y,z)$.
\item \textit{Handle-Insert-Level0$(u,v)$} : Both $u$ and $v$ are free at level $0$. After the insertion of edge $(u,v)$, procedure will include $(u,v)$ in the matching in $O(1)$ time. If $|O_u| < \sqrt{n}$, $|O_v| < \sqrt{n}$, $deg(u) < \sqrt{n}$  and $deg(v) < \sqrt{n}$ then $u$ and $v$ are removed from the free neighbour lists of their neighbours in $O(deg(u)+deg(v))$ = $O(\sqrt{n})$ time. This $O(\sqrt{n})$ computation time is associated with the creation of level $0$ epoch $(u,v)$.
\end{enumerate} 
Therefore, the computation time associated with the creation of a level $0$ epoch is $O(\sqrt{n})$ in the worst case.\\
\textbf{Termination of Level 0 epochs}:\\
An epoch in level $0$ is terminated in the body of the following procedure calls :
\begin{enumerate}
\item \textit{Naive-Settle-Augmented} : Edge $(u,v)$ is a matched edge at level $0$. Therefore, $deg(u) < \sqrt{n}$ and $deg(v) < \sqrt{n}$. Edge $(u,v)$ is removed from the matching in $O(1)$ time. Then calls are made to procedure \textit{Naive-Settle-Augmented}$(u,0)$ and \textit{Naive-Settle-Augmented}$(v,0)$. If $u$ does not have a free neighbour and $u$ is not part of any $3$ length augmenting path, then $u$ is inserted to the free neighbour list of all its neighbours. Similarly, if $v$ does not have a free neighbour and $v$ is not part of any $3$ length augmenting path, then $v$ is inserted to the free neighbour list of all its neighbours. Total computation time is $O(deg(u)+ deg(v)) = O(\sqrt{n})$. This $O(\sqrt{n})$ computation time is associated with the termination of the level $0$ epoch $(u,v)$.
\item \textit{Handle-Insert-Level0(u,v)} : Both $u$ and $v$ are at level $0$ and free. After insertion of edge $(u,v)$, procedure will include $(u,v)$ in the matching in $O(1)$ time. 
If $|O_u|$ is equal to $\sqrt{n}$ then the edge $(u,v)$ is removed from the matching in $O(1)$ time. So this terminates the level $0$ epoch $(u,v)$. This $O(1)$ computation time is associated with termination of level $0$ epoch $(u,v)$.
\item \textit{Random-Settle-Augmented(u)} : This procedure selects an edge uniformly at random from $O_u$. Let $(u,y)$ be the edge. Suppose level of $y$ is $0$, $y$ is matched and $x$ = $mate(y)$. We know that $deg(x) < \sqrt{n}$. Then the edge $(x,y)$ is removed from the matching in $O(1)$ time. 
This $O(1)$ computation time is associated with termination of level $0$ epoch $(x,y)$.
\item \textit{Random-Raise-Level-To-1(u)} : This procedure
removes the edge $(u,mate(u))$ from the matching in $O(1)$ time. This $O(1)$ computation time is associated with termination of level $0$ epoch $(u,mate(u))$. 
\item \textit{Deterministic-Raise-Level-To-1(u)} : This procedure terminates the level $0$ epoch $(u,mate(u))$ and creates the level $1$ epoch $(u,mate(u))$. However, we associate the entire computation within this procedure with the creation of level $1$ epoch $(u,mate(u))$. Therefore, computation time associated with termination the level $0$ epoch $(u,mate(u))$ is $O(1)$.  
\end{enumerate}
Therefore, computation time associated  with the termination of a level $0$ epoch is $O(\sqrt{n})$ in the worst case.  
\end{proof}
\begin{lemma}
\label{lem: Computation-Time-Level1-epochs}	
The computation time associated with the creation of a level 1 epoch is $O(n)$ and termination of a level 1 epoch is $O(n)$.
\end{lemma}
\begin{proof}
We perform an exhaustive case analysis for the creation and termination of level 1 epochs.\\
\textbf{Creation of level $1$ epochs}:\\
An epoch in level 1 is created during  the following procedure calls :
\begin{enumerate}
\item \textit{Random-Settle-Augmented}$(u)$: This procedure selects an edge $(u,y)$ uniformly at random from $O_u$. The total computation done within the procedure takes $O(deg(u)+deg(y))$. We associate this $O(deg(u)+deg(y))$ computation time with the creation of level $1$ epoch $(u,y)$. 
\item \textit{Randomised-Raise-Level-To-1$(u)$}: 
Total computation done within this procedure takes $O(deg(u))$ time. We associate this $O(deg(u))$ computation with the epoch created by the subsequent call to procedure \textit{Random-Settle-Augmented}$(u)$.  
\item \textit{Deterministic-Raise-Level-To-1$(u)$}: 
Let $u^{\prime}$ be the mate of $u$. Total computation done within this procedure takes $O(deg(u)+deg(u^{\prime}))$. We associate this $O(deg(u)+deg(u^{\prime}))$ computation time with the creation of level $1$ epoch $(u,u^{\prime})$. 
\item \textit{Fix-3-Aug-Path$(u,v,y,z)$}:
$u$ is a free vertex at level $0$, $v$ is a matched vertex at level $1$, $y$ is mate of $v$ and $z$($z \neq u$) is a free neighbour of $y$. Edge $(v,y)$ is removed from the matching and edges $(u,v)$ and $(y,z)$ are included in the matching in $O(1)$ time.
\begin{enumerate}
\item If $deg(u) \geq \sqrt{n}$ then total computation done within \textit{Fix-3-Aug-Path} takes $O(1)$ + $O(deg(z))$ = $O(deg(z))$ time.  We associate this $O(deg(z))$ computation time with the creation of level $1$ epoch $(y,z)$. 
\item If $deg(u) < \sqrt{n}$ and $deg(z) \geq \sqrt{n}$ then total computation done within \textit{Fix-3-Aug-Path} takes $O(1)$ + $O(deg(u))$ = $O(deg(u))$ time. We associate this $O(deg(u))$ computation time with the creation of level $1$ epoch $(u,v)$.
\item If $deg(u) < \sqrt{n}$ and $deg(z) < \sqrt{n}$ then total computation done within \textit{Fix-3-Aug-Path} takes $O(1)$ + $O(deg(u))$ + $O(deg(z))$ = $O(\sqrt{n})$ time. We associate this $O(\sqrt{n})$ computation time with the creation of level $1$ epoch $(u,v)$ and $(y,z)$.
\end{enumerate}
\item \textit{Fix-3-Aug-Path-D$(u,v,y,z)$}: $u$ is a free vertex at level $0$, $v$ is matched, $y$ is mate of $v$ and $z$ ($z \neq u$) is a free neighbour of $y$. Edge $(v,y)$ is removed from the matching and edges $(u,v)$ and $(y,z)$ are included in the matching in $O(1)$ time. 
\begin{enumerate}
\item Level of $v$ is $1$. Then total computation done within \textit{Fix-3-Aug-Path-D} takes $O(1)$ + $O(deg(u)+deg(z))$ = $O(deg(u)+deg(z))$ time. We associate this $O(deg(u)+deg(z))$ computation time with the creation of level $1$ epochs $(u,v)$ and $(y,z)$.
\item Level of $v$ is $0$. Then total computation done within \textit{Fix-3-Aug-Path-D} takes $O(1)$ + $O(deg(v)+ deg(y)+deg(u)+deg(z))$ time. Since $deg(v) < \sqrt{n}$ and $deg(y) < \sqrt{n}$, total computation time is $O(deg(u)+deg(z))$. We associate this $O(deg(u)+deg(z))$ computation time with the creation of level $1$ epochs $(u,v)$ and $(y,z)$.
\end{enumerate}
\end{enumerate}
Therefore, computation associated with the creation of a level $1$ epoch is $O(n)$ in the worst case.\\
\textbf{Termination of level 1 epochs}:\\
An epoch in level 1 gets terminated during the following procedure calls:
\begin{enumerate}
\item 
\label{level1-epoch-deletion-case1}
\textit{Handle-Delete-Level1}: Edge $(u,v)$ is a matched edge at level $1$.  Edge $(u,v)$ is removed from the matching in $O(1)$ time. Then calls are made to procedure \textit{Handle-Delete-Level1}$(u,0)$ and \textit{Handle-Delete-Level1}$(v,0)$. If $u$ does not get matched again then $u$ is inserted to the free neighbour list of all its neighbours. Similarly, if $v$ does not get matched again then $v$ is inserted to the free neighbour list of all its neighbours. Total computation time is $O(deg(u)+deg(v))$. This $O(deg(u)+deg(v))$ computation time is associated with the termination of level $1$ epoch $(u,v)$.
\item \textit{Random-Settle-Augmented}$(u)$: 
This procedure picks a random mate for $u$ from $O_u$. Let $y$ be the mate selected for $u$. If level of $y$ is $1$ then $y$ is matched. Let $x$ be the mate of $y$. Procedure removes $(x,y)$ from matching in $O(1)$ time. This $O(1)$ computation time is associated with termination of level $1$ epoch $(x,y)$.

\item \textit{Fix-3-Aug-Path$(u,v,y,z)$}: 
$u$ is a free vertex at level $0$, $v$ is a matched vertex at level $1$, $y$ is mate of $v$ and $z$ ($z \neq u$) is a free neighbour of $y$. Edge $(v,y)$ is removed from the matching and edges $(u,v)$ and $(y,z)$ are added to the matching in $O(1)$ time. This $O(1)$ computation time is associated with termination of level $1$ epoch $(v,y)$.
\item \textit{Fix-3-Aug-Path-D$(u,v,y,z)$}: 
$u$ is a free vertex at level $0$, $v$ is a matched vertex at level $1$, $y$ is mate of $v$ and $z$ ($z \neq u$) is a free neighbour of $y$. Edge $(v,y)$ is removed from the matching and edges $(u,v)$ and $(y,z)$ are added to the matching in $O(1)$ time. This $O(1)$ computation time is associated with termination of level $1$ epoch $(v,y)$. 
\end{enumerate}
Therefore, computation time associated with the termination of a level $1$ epoch is $O(n)$ in the worst case.
\end{proof}
\subsection{Crucial Classification of Level 1 epochs}
\label{subsubsec: Random-Deterministic-Level1-epochs}
We classify  level 1 epochs into two categories : random level 1 epochs and deterministic level 1 epochs.  We consider this as a novel step in the extension of the analysis technique of \cite{DBLP:BGS}.
\begin{definition}
\label{def:random-level1-epoch}
Random Level 1 Epoch :
Let us consider a level 1 epoch $(u,v)$. Without loss of generality, suppose this was created due to vertex $u$. At the time of creation of the epoch, if $(u,v)$ was selected uniformly at random from $O_u$ then the epoch $(u,v)$ is called a random level 1 epoch.  
\end{definition}
\begin{definition}
\label{def:deterministic-level1-epoch}
Deterministic Level 1 Epoch : Let us consider a level 1 epoch $(x,y)$. Without loss of generality, suppose this was created due to $x$. At the time of creation, if the mate of $x$ is deterministically chosen to be $y$ then epoch $(x,y)$ is called deterministic level 1 epoch.
\end{definition}
Therefore, over any sequence of updates,  each level 1 epoch is either a random level 1 epoch or a deterministic level 1 epoch. We now further refine the classification of deterministic level 1 epochs based on the time to taken for their creation.  In particular, we consider deterministic level 1 epochs based which are created in $O(\sqrt{n})$ time.  \\
{\bf Inexpensive Deterministic Level 1 Epochs:} From Table \ref{table : Procedure_Calls_VS_Epocs}, the two procedures that could create deterministic level 1 epochs in $O(\sqrt{n})$ time are \textit{Fix-3-Aug-Path} and \textit{Fix-3-Aug-Path-D}.
Consider the call to procedure \textit{Fix-3-Aug-Path}$(u,v,y,z)$ where $u$ is a free vertex at level $0$ and $v$ is a matched vertex at level 1. Let $y$ be mate of $v$ and $z$ be a free neighbour of $y$. If $deg(u) < \sqrt{n}$ and $deg(z) < \sqrt{n}$ then procedure will terminate epoch $(v,y)$ and create deterministic level 1 epochs $(u,v)$ and $(z,y)$. Further, since $deg(u) < \sqrt{n}$ and $deg(z) < \sqrt{n}$, from Lemma~\ref{lem: Computation-Time-Level1-epochs}, the computation involved in the creation of these epochs takes time $O(\sqrt{n})$.    We refer to the epochs created by \textit{Fix-3-Aug-Path}$(u,v,y,z)$ as \textit{inexpensive deterministic level 1 epochs}.

Secondly,  a call to procedure \textit{Fix-3-Aug-Path-D} may also create a deterministic level $1$ epoch in time $O(\sqrt{n})$. However, from Observation~\ref{obs : flag-change-0-to-1}, a call to procedure \textit{Fix-3-Aug-Path-D} is always preceded by a call to procedure \textit{Random-Settle-Augmented}. Therefore, the time associated with epochs created by procedure \textit{Fix-3-Aug-Path-D} are analyzed in an epoch-set(Definition~\ref{def:epoch-set}) whose representative is the epoch created by the preceding \textit{Random-Settle-Augmented}.
 Therefore, for the rest of the discussion we will consider the epochs created by  \textit{Fix-3-Aug-Path} as the only inexpensive deterministic epochs.

If the termination of an \textit{inexpensive deterministic level 1 epoch} also takes $O(\sqrt{n})$ time then total computation time of the epoch is $O(\sqrt{n})$. However, if termination of an \textit{inexpensive deterministic level 1 epoch} involves the computation time due to $\Omega(\sqrt{n})$ edges,  then we need a careful accounting for the total computation time.   For this we classify 
inexpensive deterministic level 1 epochs into two types- type 1 and type 2.  \\
{\bf Type 1 inexpensive deterministic level 1 epoch:} Let $(u,v)$ be an inexpensive deterministic level 1 epoch which was created by $u$, where $deg(u) <  \sqrt{n}$ at the time of creation.  The epoch is defined to be of type 1 if it satisfies one of the following conditions:
\begin{enumerate}
\item On termination $deg(u)$ is $\Omega(\sqrt{n})$.
\item On termination $deg(u) < \sqrt{n}$ and $deg(v)$ is $\Omega(\sqrt{n})$, and on creation  $deg(v) < \sqrt{n}$.
\item On termination $deg(u) < \sqrt{n}$ and $deg(v) < \sqrt{n}$, and on creation $deg(v)$ is $\Omega(\sqrt{n})$ . 
\item On termination $deg(u) < \sqrt{n}$ and $deg(v)$ is $\Omega(\sqrt{n})$, and on creation $deg(v)$ is $\Omega(\sqrt{n})$, and when $v$ is matched for the first time after termination of the epoch $(u,v)$ by the algorithm, $deg(v) < \sqrt{n}$.
\item On termination $deg(u) < \sqrt{n}$ and $deg(v)$ is $\Omega(\sqrt{n})$, and on creation $deg(v)$ is $\Omega(\sqrt{n})$, and after this $v$ remains unmatched for some number of updates after which for the first time $deg(v) < \sqrt{n}$.  Note that in this case, we crucially use the fact that the extended update sequence all the edges are eventually deleted.  Therefore, for each vertex at some point in the update sequence, the degree will be less than $\sqrt{n}$.
\end{enumerate}	
{\bf Type 2 inexpensive deterministic level 1 epoch:} Let $(u,v)$ be an inexpensive deterministic level 1 epoch which was created by $u$, where $deg(u) < \sqrt{n}$ at the time of creation.  The epoch is defined to be of type 2 if it satisfies the following conditions
\begin{itemize}
\item On creation $deg(v)$ is $\Omega(\sqrt{n})$, on termination $deg(u)$ is less than $\sqrt{n}$ and $deg(v)$ is $\Omega(\sqrt{n})$, and when $v$ is matched first time after the termination of epoch $(u,v)$, $deg(v)$ is $\Omega(\sqrt{n})$.
\end{itemize}	
\begin{lemma}
\label{lem:inexpensive-level-1-epoch-termination}	
For each type 1 inexpensive deterministic level 1 epoch $(u,v)$, there exists a set of  $\Omega(\sqrt{n})$ many $O(1)$ time updates which involve either $u$ or $v$.  Further,tThese $\Omega(\sqrt{n})$ many $O(1)$ time updates are those that occur after the creation of the epoch $(u,v)$, are associated with the epoch $(u,v)$ only, and satisfy one of the following conditions:
\begin{itemize}
\item They occur before the termination of epoch $(u,v)$ when degree of $u$ or $v$ is $\Omega(\sqrt{n})$.
\item They occur before the termination of the epoch $(u,v)$ when the degree of both $u$ and $v$ is $O(\sqrt{n})$, but on creation of epoch $(u,v)$, $deg(v)$ is $\Omega(\sqrt{n})$.
\item They occur before the first time $v$ is matched after the termination of the epoch $(u,v)$.
\item They occur before the first time $deg(v)$ becomes less than $\sqrt{n}$ after the termination of epoch $(u,v)$.
\end{itemize}
 Consequently, type 1 inexpensive deterministic level 1 epochs  contribute  $O(\sqrt{n})$ to the amortized update time. 
\end{lemma}
\begin{proof}
To prove the claim, we  analyse the five different cases in the definition of  type 1 inexpensive deterministic epochs as follows, and at the end of each case the common proposition is that the $\Omega(\sqrt{n})$ many $O(1)$ time updates are all associated with only the epoch $(u,v)$ and not with any other epoch :
\begin{enumerate}
\item On termination of epoch $(u,v)$, $deg(u)$ is $\Omega(\sqrt{n})$: At the time of creation of epoch $(u,v)$, $deg(u)$ was strictly less than $\sqrt{n}$. Therefore, from the time of creation of epoch $(u,v)$ till its termination there is at least $\sqrt{n}$  many insertions have taken place at vertex $u$ during the epoch $(u,v)$. All these insertions are processed in $O(1)$ time because $(u,v)$ is already in the matching.
\item On termination $deg(u)$ is less than $\sqrt{n}$, on creation  $deg(v) < \sqrt{n}$ and on termination $deg(v)$ is $\Omega(\sqrt{n})$: In this case  $\Omega(\sqrt{n})$ many insertions would have taken place at $v$ during the epoch $(u,v)$, and these would have been processed in $O(1)$ time each as $(u,v)$ is already in the matching.  
\item On termination $deg(u)$ is less than $\sqrt{n}$, on creation $deg(v)$ is $\Omega(\sqrt{n})$ and on termination $deg(v) < \sqrt{n}$: In this case there would have been  $\Omega(\sqrt{n})$ edges deleted at $v$ during the epoch $(u,v)$ and each of them would have been processed in $O(1)$ time since $(u,v)$ is already in the matching.  
\item On termination $deg(u)$ is less than $\sqrt{n}$, on creation $deg(v)$ is $\Omega(\sqrt{n})$ and on termination $deg(v)$ is $\Omega(\sqrt{n})$ and when $v$ is matched first time again by the algorithm $deg(v)$ is less than $\sqrt{n}$: In this case, after termination of epoch $(u,v)$ till $v$ gets matched again for the first time, it would have become a level 0 vertex. Since $deg(v)$ at this time is less than $\sqrt{n}$, there would have been  $\Omega(\sqrt{n})$ deletes of edges incident on the free vertex $v$.  
\item On termination $deg(u)$ is less than $\sqrt{n}$ and $deg(v)$ is $\Omega(\sqrt{n})$, and on creation $deg(v)$ is $\Omega(\sqrt{n})$, and $v$ remains unmatched for some number of updates after which $deg(v) < \sqrt{n}$: As in the preceding case,
after termination of epoch $(u,v)$ since $v$ remains unmatched till the first time when $deg(v) < \sqrt{n}$. During these updates it would have become a level 0 vertex. Since $deg(v)$ at this time is less than $\sqrt{n}$, there would have been $\Omega(\sqrt{n})$ deletes of edges incident on the free vertex $v$.  
\end{enumerate}
Therefore, we have proved that for a type 1 inexpensive deterministic level 1 epoch $(u,v)$ for which termination takes $\Omega(\sqrt{n})$ time, there exist a set of $\Omega(\sqrt{n})$ many updates of edges incident or $u$ or $v$ which is processed in $O(1)$ time.    Further, in each of the cases considered we have also proved the second statement of the lemma.
Finally,  The total time associated with the creation and termination of the epoch $(u,v)$ is $O(n)$.  Since  the epoch $(u,v)$ is associated with  a set of $\Omega(\sqrt{n})$ many constant time updates of edges incident or $u$ or $v$, the contribution to the amortized cost by the epoch  $(u,v)$  $O(\sqrt{n})$. Hence the lemma is proved.
\end{proof}
 Lemma~\ref{lem:inexpensive-level-1-epoch-termination} shows that we can distribute the $O(n)$ cost of the termination of the inexpensive deterministic level 1 epoch to $\Omega(\sqrt{n})$ many  $O(1)$ time updates which happen during the epoch $(u,v)$.
For the rest of the analysis we will not consider the type 1 deterministic level 1 epochs which requires time $O(\sqrt{n})$ for creation. 
Further, by definition,  the creation time associated with type 2 epochs is inexpensive, but termination and subsequent rematching is expensive.
We will be interested only  type 2 inexpensive deterministic level 1 epochs and those deterministic level 1 epochs whose creation takes $\Omega(\sqrt{n})$ time. Therefore, quite naturally, we refer to these as {\em expensive deterministic level 1 epochs}.
 \subsection{Epoch-Sets:Grouping expensive deterministic level 1 epochs with Random level 1 epochs}
We first show that each type 2 inexpensive deterministic level 1 epoch $(u,v)$ can be associated with the creation of random level 1 epoch involving $v$ when it gets matched again for the first time after the termination of the epoch $(u,v)$. Intuitively, this
time taken to match $v$ again for the first time after the deletion of $(u,v)$ from the matching is used to account for the time taken to delete $(u,v)$ from the matching, which in this case is expensive at the time of termination.  Note, that the deletion of $(u,v)$ from the matching and the subsequent re-matching of $v$ could be during different updates.  
\begin{lemma}
\label{lem:type2epochs}
Let $(u,v)$ be a type 2 inexpensive deterministic level 1 epoch  created by $u$.  When $v$ is matched again for the first time after the termination of the epoch $(u,v)$ it creates a random level 1 epoch.  Further, for each random level 1 epoch at $v$, there is at most one type 2 inexpensive deterministic level 1 epoch involving a matching edge containing $v$.  
 \end{lemma}
\begin{proof}
Let $(u,v)$ be a type 2 inexpensive deterministic level 1 epoch  created by $u$, on creation $deg(u) < \sqrt{n}$ and $deg(v)$ is $\Omega(\sqrt{n})$,  and on termination of the epoch $(u,v)$, $deg(u)$ is less than $\sqrt{n}$ and $deg(v)$ is $\Omega(\sqrt{n})$, and when $v$ is matched first time again by the algorithm, $deg(v)$ is $\Omega(\sqrt{n})$. 
We prove the claim by considering the following two cases.   We first consider the case in which $v$ is matched again in the same update step in which $(u,v)$ is terminated. In this case, \textit{Handle-Delete-Level1}$(v,flag)$ would have been called, and since $deg(v)$ is $\Omega(\sqrt{n})$, it would have called \textit{Random-Settle-Augmented}$(v)$ and subsequently due to the fact that $v$ gets matched again, and it has a degree of  $\Omega(\sqrt{n})$, it creates a random level 1 epoch at $v$.  In the second case, if $v$ does not get matched in the same update step in which the epoch $(u,v)$ is terminated, then it would have become a free vertex at level 0.  When it becomes matched at level 0 for the first time after the termination of the epoch $(u,v)$, due to its degree being $\Omega(\sqrt{n})$ at the time of it getting matched, \textit{Randomised-Raise-Level-to-1}$(v)$ is called, and this creates a random epoch at level 1.  

The second statement is true due to the following reason: For each type 2 inexpensive deterministic level 1 epoch $(u,v)$ created by $u$, there is at most one {\em first} time instant at which $v$ is matched again.  Hence the lemma is proved.
\end{proof}
As in Lemma \ref{lem:type2epochs} we now show that the remaining expensive deterministic level 1 epochs, that is those that take  $\Omega(\sqrt{n})$ time for creation, are preceded by the creation of a random level 1 epoch. Again, as in the previous lemma this paves the way for accounting  the contribution of the expensive deterministic level 1 epoch to the total update time.  The crucial difference is that the expensive deterministic epoch and the preceding random level 1 epoch are created during the same update step.
\begin{lemma}
\label{lem:mapping-deterministic-to-random-epochs}	
Creation of a deterministic level 1 epoch during an update which takes $\Omega(\sqrt{n})$ time is preceded by the creation of a random level 1 epoch during the same update. 
\end{lemma}
\begin{proof}
From Table~\ref{table : Procedure_Calls_VS_Epocs} and Lemma~\ref{lem: Computation-Time-Level1-epochs}, the following procedure calls are the only calls which create a deterministic level 1 epoch where the creation time is $\Omega(\sqrt{n})$ time:
\begin{enumerate}
\item \textit{Deterministic-Raise-Level-To-1}: A call to procedure \textit{Deterministic-Raise-Level-To-1} creates a deterministic level 1 epoch. 
From Table~\ref{table : Procedure_Calls} this call is made in \textit{Naive-Settle- Augmented} with $flag$ value 1 during the same update.  Since $flag$ has value 1,  by applying observation~\ref{obs : flag-change-0-to-1} we conclude that there must have been a preceding call to \textit{Random-Settle-Augmented} which from Lemma~\ref{lem: Computation-Time-Level1-epochs} and Table~\ref{table : Procedure_Calls_VS_Epocs} creates a random level 1 epoch. 
\item Fix-3-Aug-Path-D: A call to procedure \textit{Fix-3-Aug-Path-D} creates two deterministic level 1 epochs.  Again, from Table~\ref{table : Procedure_Calls} this call is either made in \textit{Random-Settle-Augmented} or in \textit{Naive-Settle- Augmented} with $flag$ value 1 during the same update.  If the calling procedure was procedure \textit{Random-Settle-Augmented}, then it creates a random level 1 epoch.  In the case when the calling procedure is  \textit{Naive-Settle- Augmented} with $flag$ value 1,
by applying observation~\ref{obs : flag-change-0-to-1} we conclude that there must have been a preceding call to \textit{Random-Settle-Augmented}.  In either case, there is a preceding call to \textit{Random-Settle-Augmented} during the same update as the \textit{Fix-3-Aug-Path-D}, and from Lemma~\ref{lem: Computation-Time-Level1-epochs} and Table~\ref{table : Procedure_Calls_VS_Epocs} it creates a random level 1 epoch. 
\end{enumerate}
Hence the Lemma.
\end{proof}
From Lemma \ref{lem:type2epochs}, it is clear that every random level 1 epoch is associated with at most one preceding type 2 inexpensive deterministic level 1 epoch. Further, 
from Lemma~\ref{lem:mapping-deterministic-to-random-epochs}, it is clear that every deterministic level 1 epoch whose creation takes $\Omega(\sqrt{n})$ time is associated with the creation of a  random level 1 epoch preceding it in the same update step. In the following, we group all the such deterministic epochs which are associated with the same random level 1 epoch into one set which we refer to as the \textbf{\textit{epoch-set}}.
\begin{definition}
\label{def:epoch-set}
An epoch-set  is defined for each random level 1 epoch $(u,v)$ created at $u$ during an update by a procedure, say $P$, and is denoted by $\xi_{u}$.  The epoch-set for the random level 1 epoch  $(u,v)$ consists of the subsequent deterministic level 1 epochs, if any, created in procedures called from $P$, during the same update, and before the creation of the next random level 1 epoch during the update step. Further, if there is a type 2 inexpensive deterministic level 2 epoch $(x,u)$ after which $u$ was free till this update in which epoch $(u,v)$ is created, then $\xi_{u}$ contains the epoch $(x,u)$.  The random level 1 epoch in an epoch-set is referred to as its representative.  In this case, $(u,v)$ is the representative of $\xi_{u}$. 
\end{definition}
\begin{lemma}
\label{lem:cardinality of epoch-set}	
The number of elements in an epoch-set is at most 63.
\end{lemma}
\begin{proof}
Let $(u,v)$ be the representative of an epoch-set $\xi_{u}$. By definition, $\xi_{u}$ consists of some expensive deterministic level 1 epochs that follow the creation of $(u,v)$  during the same update, and at most one another preceding epoch that involves $u$.  Further, by Lemma~\ref{lem: Computation-Time-Level0-epochs}, Lemma~\ref{lem: Computation-Time-Level1-epochs} and Table~\ref{table : Procedure_Calls_VS_Epocs} each epoch is created by a procedure call during the update.  From the description of the procedures in Section~\ref{subsec : Description_Procedures},
we know that each procedure introduces at most two edges into the matching, and thus creates at most two epochs.  
From Lemma~\ref{thm:number-of-calls-perupdate} we know that the number of procedure calls made on an insert or delete is at most 31.  Therefore, the number of epochs in an epoch-set is at most 63.   
\end{proof}
\begin{lemma}
\label{lem:total-computation-time-epoch-set}	
Total computation time associated with an epoch-set i.e. time required for creation and termination of all the epochs in the epoch-set is $O(n)$.
\end{lemma}
\begin{proof}
From Lemma~\ref{lem: Computation-Time-Level1-epochs}, creation and termination time for level 1 epochs is $O(n)$.
From Lemma \ref{lem:cardinality of epoch-set},  it is clear that number of level 1 epochs associated with an epoch-set is at most 63. Therefore, total computation associated with an epoch-set which is the time for creation and termination of all the epochs in the epoch-set is $O(n)$.
\end{proof}
\subsection{Expected value of total update time}
\label{subsec : expected-total-update-time}
We complete our analysis here by placing an upper bound on the expectation of $T$, the total time taken by the algorithm to service a sequence of updates.  The analysis is presented in the following paragraphs and it leads to the proof of Theorem \ref{thm : expected-total-update-time}.  At the beginning of  Section \ref{epochanalysis}, we saw that $T$ is written as the sum of time taken to create and terminate the different epochs associated with each of the updates.  We now bound the contribution of the level 0 epochs, the type 1 inexpensive deterministic level 1 epochs, and
finally the contribution of the epoch-sets to the expected value of $T$.  

\noindent
From Lemma~\ref{lem: Computation-Time-Level0-epochs}, total computation associated with each of the level $0$ epochs is $O(\sqrt{n})$. Therefore, during any sequence of updates, if there are $t_1$ level $0$ epochs then computation associated with all the epochs is deterministically bounded by $O(t_1 \sqrt{n})$.  Next, we consider inexpensive deterministic level 1 epochs, that is those whose creation take $O(\sqrt{n})$ time.  Among these epochs, if there are $t_2$ of them which also are terminated in $O(\sqrt{n})$ time, then their contribution to $T$ is bounded by $O(t_2 \sqrt{n})$.  The amortized cost per operation in either of these two cases is $O(\sqrt{n})$.  
Following this,  Lemma~\ref{lem:inexpensive-level-1-epoch-termination} shows that for a type 1 inexpensive level 1 deterministic epoch $(u,v)$
there exists a set of $\Omega(\sqrt{n})$ many $O(1)$ time updates involving $u$ or $v$ which are associated with the epoch $(u,v)$ and with no other epoch.   Therefore, during a sequence of updates, if there are $t_3$ type 1 inexpensive deterministic level 1 epochs whose termination results in $O(n)$ time computation then there exist at least $t_3 \Omega(\sqrt{n})$ many updates each of which takes deterministic $O(1)$ time.  Let $t'_{3}$ the total number of such updates.  Therefore, the contribution to $T$ from these updates is $t_3 n + t'_3$.  Therefore, the amortized cost per operation is at most $\frac{t_3 n + t'_3}{t'_3}$. Since $t'_3 \geq t_3 \sqrt{n}$, it follows that the amortized cost of these operations is $O(\sqrt{n})$. 

The only remaining epochs to be accounted for are the expensive deterministic level 1 epochs and random level 1 epochs.  
From  Definition \ref{def:epoch-set} and Lemma \ref{lem:cardinality of epoch-set}, it follows that we have a partition of these remaining  epochs into epoch-sets of constant size.  From Lemma \ref{lem:total-computation-time-epoch-set}	we know that each epoch-set contributes an $O(n)$ term to the total update time.
  Every such epoch-set has a representative element which is a random level 1 epoch. We now bound the contribution by the expected  running time of these epoch-sets to  the expected value of $T$.  We do this by setting up a random variable similar to the one in the analysis by Baswana et al. \cite{DBLP:BGS}.   
  \subsubsection*{Expected contribution of the epoch-sets to the expected total update time}
Let $X_{v,k}$ be a random variable which is 1 if $v$ creates a random level 1 epoch at update step $k$, otherwise $X_{v,k}$ is set to $0$. 
Let $O_v^{init}$ denote the set of edges owned by $v$ at the time of creation of the random level 1 epoch at $v$.   Let $Z_{v,k}$ denote the number of edges deleted from $O_v^{init}$ before the deletion of the edge corresponding to the random level 1 epoch
from the graph (recall, that deletions from the graph happen only by updates).   Further, we define that if $X_{v,k}$ = $0$, then $Z_{v,k}=0$.
From the description of the algorithm in Section \ref{algorithm}, a random level 1 epoch is created only by \textit{Random-Settle-Augmented}.  The condition at the time of the creation of the epoch is  $|O_v^{init}| \geq \sqrt{n}$, and this is used crucially in the Lemma~\ref{lem : Expected_Deletion}.  Consequently, for any sequence of $t_4$ updates, $\displaystyle \sum_{v,k}^{}Z_{v,k} \leq t_4$. Hence  $\displaystyle \sum_{v,k}^{} \mathbb{E}(Z_{v,k}) \leq t_4$.
\begin{lemma}
\label{lem : Expected_Deletion}	
For each vertex $v$ and integer $k \geq 0$, 
	$\mathbb{E}(Z_{v,k}) \geq \frac{\sqrt{n}}{2}\cdot \mathbb{P}r(X_{v,k} = 1)$.
\end{lemma}
\begin{proof}
Consider the expected value $\mathbb{E}(Z_{v,k} | X_{v,k} = 1)$ which is the expected value of $Z_{v,k}$ conditioned on the event that $v$ creates random level 1 epoch at update  $k$. Let $O_v^{init}$ be the number of edges owned by $v$ at the time of creation of the random level 1 epoch at by $v$. The value of $Z_{v,k}$ depends on the deletion sequence and the random choice of mate of vertex $v$.  We know that the mate of $v$ is distributed uniformly over $O_v^{init}$ and  that all the edges are deleted eventually in our extended update sequence.
Therefore, $\mathbb{E}(Z_{v,k} | X_{v,k} = 1) = \frac{ |O_v^{init}|}{2} \geq \frac{\sqrt{n}}{2}$.  Consequently, 
\begin{equation}
\label{eq : 1}
\mathbb{E}(Z_{v,k}) = \mathbb{E}(Z_{v,k} | X_{v,k} = 1)\cdot \mathbb{P}r(X_{v,k} = 1) \geq \frac{\sqrt{n}}{2} \cdot \mathbb{P}r(X_{v,k} = 1)
\end{equation}
\end{proof}
{\bf Bounding the Expected Cost:}
From Lemma~\ref{lem:total-computation-time-epoch-set}, the total computation associated with an epoch-set is $O(n)$.   From Lemma \ref{thm:number-of-calls-perupdate}, the number of procedure calls per update is at most a constant.   Therefore, the number of epoch-sets created per update is a constant. Consequently, the computation cost associated with each update  is $C \cdot n$ for some constant $C$. Hence expected value of the total update time  a sequence of $t_4$ updates which create random level 1 epochs is :
\begin{equation}
\label{eq:expected-update}
\begin{split}
\sum_{v,k}^{} C \cdot n \cdot \mathbb{P}r(X_{v,k} = 1) & = 2\cdot C \cdot \sqrt{n} \sum_{v,k}^{} \frac{\sqrt{n}}{2} \cdot \mathbb{P}r(X_{v,k} = 1) \\
& \leq 2 \cdot C \cdot\sqrt{n} \sum_{v,k}^{}  \mathbb{E}(Z_{v,k}) \;\; (\text{using Equation } \; \ref{eq : 1})\\
& \leq 2\cdot C\cdot\sqrt{n}\cdot t_4  \nonumber
\end{split}
\end{equation}
Consequently, for any sequence of $t_4$ updates which create a random level 1 epoch, the contribution to the expected total update time is $O(t_4 \sqrt{n})$. This gives us the following theorem which is the first main result of this paper.
\begin{theorem}
\label{thm : expected-total-update-time}	
Starting with a graph on $n$ vertices and no edges for any sequence of $t$ fully dynamic update operations, our data structure maintains a maximal matching after removing all augmenting paths of length at most 3 (consequently, a $3/2$ approximate maximum cardinality matching) at the end of each update in expected total update time $O(t\sqrt{n})$.	
\end{theorem} 
\begin{proof}
Let $t'$ be the length of the extended update sequence which has been considered in the analysis preceding this theorem.  We know that $t' \leq 2t$.  In the discussion preceding the statement of the theorem, we have  proved that the expected total update time on the extended sequence is $O(t' \cdot \sqrt{n})$.  Therefore, from Observation \ref{extendupdate} we know that the expected total update time on the update sequence of $t$ updates is $O(t' \cdot \sqrt{n})$, which in turn is
$O(t \cdot \sqrt{n})$ since $t' \leq 2t$.
\end{proof}
\subsection{Worst case total update time with high probability}
Here too we analyze the worst case total update time for the extended update sequence of length $t' \leq 2t$ which guarantees that at the end all the edges are deleted.  We  classify the epoch-sets into two categories: good epoch-sets and  bad epoch-sets.
\label{subsec:worstcase-total-updatetime}
\begin{definition}
	\label{def: bad-epoch-set}	
	Bad epoch-set : An epoch-set $\xi_u$ is said to be bad if the representative epoch $(u,v)$ is terminated  before the deletion from the graph of the first $\frac{1}{3}$ edges that $u$ owned at the time of creation of the epoch $(u,v)$.
	An epoch-set is good if it is not bad.  Intuitively, if $\xi_u$ is a good epoch-set then the ownership list of $u$ undergoes many (at least $\frac{\sqrt{n}}{3}$) deletions before deletion of $(u,v)$.
\end{definition} 
\begin{lemma}
	\label{lem:epoch-set is bad}	
	Suppose vertex $v$ is the vertex that created the epoch corresponding to the  representative of the epoch-set $\xi_{v}$  during the $k^{th}$ update for $k \leq t'$. Then the probability that the epoch-set is a bad epoch-set is at most  $\frac{1}{3}$.  
\end{lemma}
\begin{proof}
	Let $O_v^{init}$ denote the ownership list of $v$ when $v$ created the representative epoch of $\xi_v$.   Let us consider the sequence $D$ of edge deletion updates to the graph which delete the edges in  $O_v^{init}$.  Whether the epoch-set $\xi_{v}$ is bad or good is fully determined by the mate of $v$ picked at the time of creation of the representative epoch of the epoch-set $\xi_{v}$ and the sequence $D$. The epoch-set is bad if the mate of $v$ is among the endpoints of the first $\frac{O_v^{init}}{3}$ edges in $D$. We know that the mate of $v$ is equally likely to be any edge from the set $O_v^{init}$, since the representative epoch is selected at random from $O_v^{init}$. Therefore, the probability that the epoch-set $\xi_{v}$ is bad is at most $\frac{1}{3}$.
\end{proof}
\begin{lemma}
\label{lem: bad-vs-good-epoch-set}	
At the end of execution of our algorithm for a given sequence of updates, the probability that the number of bad epoch-sets exceeds the number of good epoch-sets by $i$ is at most $\frac{1}{2^i}$.	 Therefore, at the end of executionthe number of bad epoch-sets exceeds the number of good epoch-sets by $2 \log_2 n$ with probability at most $\frac{1}{n^2}$.
\end{lemma}
\begin{proof}
Suppose at the end of execution of our algorithm, the number of bad epoch-sets is at least $r+i$ and the number of good epoch-sets is at most $r$ and let us call it as event $\mathbb{A}$. Let us consider the event when the number of bad epoch-sets is equal to $k+i$ and the number of good epoch-sets is $k$ for some positive integer $k$ at some step during the run of the algorithm.   Let $\mathbb{B}$ denote this event.  If event $\mathbb{A}$ happens, since initially there are zero  good epoch-sets and zero bad epoch-sets, and 
 each edge, say $(u,v)$, is owned by exactly one of its end points, exactly one of $u$ and $v$ can randomly select $(u,v)$ as the representative for the epoch-set created. Therefore, 
at the end of the update sequence each epoch-set which is created is either a good epoch-set or a bad epoch-set. Therefore,  it follows that there must be an update during which the
number of good epoch-sets is equal to $k$ and the number of bad epoch-sets is equal to $k+i$, for some $0 \leq k \leq r$.  Therefore,  $\mathbb{P}r(\mathbb{A}) \leq \mathbb{P}r(\mathbb{B})$.  
Let $p$ denote the probability that an epoch-set is good and $q$ denote the probability that an epoch-set is bad. Therefore, $q = 1 - p$. From Lemma~\ref{lem:epoch-set is bad}, $q \leq \frac{1}{3}$ and $p \geq \frac{2}{3}$. Therefore, probability of occurrence of event $\mathbb{A}$ is obtained by finding an upper bound on the probability of event $\mathbb{B}$ as follows :
\begin{equation}
\begin{split}
\mathbb{P}r(\mathbb{A})  \leq \mathbb{P}r(\mathbb{B}) 
& = \binom{2k+i}{k} p^k \cdot q^{k+i} = \binom{2k+i}{k}p^{k+i}\cdot q^{k}\cdot{\frac{q^i}{p^i}} \\
& \leq (p+q)^{2k+i}\cdot{\frac{q^i}{p^i}} = {\frac{q^i}{p^i}} \leq {\frac{1}{2^i}}
\end{split}
\end{equation}
Therefore, for $i = 2 \log_2 n$, $\mathbb{P}r(\mathbb{A}) \leq \frac{1}{n^2}$. Therefore, at the end of execution of our algorithm for any given sequence of updates, the number of bad epoch-sets exceeds the number of good epoch-sets by $2 \log_2 n$ with probability at most $\frac{1}{n^2}$.\\  
\end{proof}

\begin{lemma}
\label{lem: Number-of-good-epochsets}
Over any sequence of $t'$ updates number of good epoch-sets created is at most $\frac{3t'}{\sqrt{n}}$.
\end{lemma}
\begin{proof}
Let $\xi_v$ be a good epoch-set and let $(u,v)$ be its representative.  Let $O_v^{init}$ be the set of edges owned by $v$ when the epoch $(u,v)$ is created.  
By the definition of a good epoch-set, it follows that the epoch $(u,v)$ does not get terminated before at least $\frac{O_v^{init}}{3}$ of the edges in $O_v^{init}$ are deleted from the graph.  We know that $O_v^{init}$ has size at least $\sqrt{n}$.  Further, from Lemma~\ref{lem: ownership_edge} each edge is owned by exactly one of its end-points.  Therefore,  it follows that the number of updates is at least the number of good epoch-sets times $\frac{\sqrt{n}}{3}$.  Therefore, the number of good-epoch sets is at most $\frac{3t'}{\sqrt{n}}$.  Hence the lemma. 
%
%
\end{proof}
From Lemma~\ref{lem: bad-vs-good-epoch-set} and Lemma~\ref{lem: Number-of-good-epochsets}, over any sequence of $t'$ updates the number of bad epoch-sets created is at most $\frac{3t'}{\sqrt{n}} + 2\log_2 n$ with probability $ \geq 1 - \frac{1}{n^2}$. From Lemma~\ref{lem:total-computation-time-epoch-set}, total computation associated with an epoch-set is $C\cdot n$ for some constant $C$. Therefore, over any sequence of $t'$ updates, total computation time taken by our algorithm in the worst case with probability at least  $1 - \frac{1}{n^2}$ is :\\
$\frac{3t'}{\sqrt{n}} \cdot C\cdot n$ + $(\frac{3t'}{\sqrt{n}}+ 2 \log_2 n) \cdot C\cdot n$ = $2 \cdot C \cdot(3t' \cdot \sqrt{n} + n \log_2 n)$ \\\\
Therefore, for any sequence of $t' \geq {\sqrt{n}}\log_2 n$ updates, the amortized update time of our algorithm is $O(\sqrt{n})$ with probability at least $1 - \frac{1}{n^2}$.  We conclude with the following theorem  which is the second main result in this paper.  The proof follows on the same lines as the the proof of Theorem \ref{thm : expected-total-update-time} using Observation \ref{extendupdate}.
\begin{theorem}
\label{thm : worstcase-total-update-time}	
Starting with a graph on $n$ vertices and no edges and ending with a graph with no edges, our data structure maintains a maximal matching after removing all augmenting paths of length at most 3 (consequently, a $3/2$ approximate maximum cardinality matching) at the end of each update for any sequence of $t$ update operations in $O(t\sqrt{n} + n\log n)$ time with high probability.	
\end{theorem} 

{\bf Acknowledgement:} We thank Sumedh Tirodkar for pointing out an error in our earlier draft.
\bibliography{references.bib}
\section{Appendix}
\label{sec : Appendix}
\begin{algorithm}[h]
	\textit{Removed-flag}$\leftarrow 0$;\\
	$y \leftarrow$ mate$(v)$;\\
	\uIf{$u \in F(y)$}{
		Delete $u$ from $F(y)$;\\
		\textit{Removed-flag}$\leftarrow 1$;\\
	}
	\uIf{\textbf{has-free}(y)}{
		$z \leftarrow$ \textbf{get-free}$(y)$;\\
		\uIf{\textit{Removed-flag} == $1$}{
			Add $u$ to $F(y)$;\\
		}
		return $z$;
	}\Else{
	\uIf{\textit{Removed-flag} == $1$}{
		Add $u$ to $F(y)$;\\
	}
	return NULL;
}			 
\caption{Check-3-Aug-Path($u,v$)}
\label{alg:Check-3-Aug-Path}
\end{algorithm}
\begin{algorithm}[h]
	\For{every $(u,w) \in O_u$ and Level$(w)$ == 1}{
		Move $(u,w)$ from $O_{u}$ to $O_w$;
	} 
	\caption{Transfer-Ownership-From($u$)}
	\label{alg:Transfer-Ownership-From}
\end{algorithm}
\begin{algorithm}[h]
	\ForEach{$w \in N(u)$}{
		\uIf{Level$(w)$ == 0}{
			\uIf{$(v,w) \in O_w$}{
				Remove $(u,w)$ from $O_w$;\\
				Add $(u,w)$ to $O_{u}$;\\
			} 
		}	 	
	} 
	\caption{Transfer-Ownership-To($u$)}
	\label{alg:Transfer-Ownership-To}
\end{algorithm}
\begin{algorithm}[h]
	\ForEach{$w \in N(u)$}{
		\uIf{$(u,w) \in O_w$}{
			Remove $(u,w)$ from $O_w$;\\
			Add $(u,w)$ to $O_u$;\\
		}	 	
	} 
	\caption{Take-Ownership($u$)}
	\label{alg:Take-Ownership}
\end{algorithm}
\begin{algorithm}[h]
	\ForEach{$w \in N(u)$}{
		Insert $u$ to $F(w)$;	 	
	} 
	\caption{Insert-To-F-List($u$)}
	\label{alg:Insert-To-F-List}
\end{algorithm}
\begin{algorithm}[h]
	\ForEach{$w \in N(u)$}{
		Delete $u$ from $F(w)$;	 	
	} 
	\caption{Delete-From-F-List($u$)}
	\label{alg:Delete-From-F-List}
\end{algorithm}

\begin{algorithm}[h]
\uIf{\textbf{has-free}(u)}{
		$w \leftarrow$ \textbf{get-free}(u);\\
		$M \leftarrow M \cup \{(u,w)\}$;\\
		\uIf{$deg(u) \geq \sqrt{n}$}{
			\uIf{flag = 1}{
				\textbf{Deterministic-Raise-Level-To-1}$(u)$;\\
				\textbf{Delete-From-F-List}$(u)$; \textbf{Delete-From-F-List}$(w)$; \\
			}\Else{
			\textbf{Randomised-Raise-Level-To-1}$(u)$;\\
			
		}	
	}\Else{
	\uIf{$deg(w) \geq \sqrt{n}$}{
		\uIf{flag = 1}{
			\textbf{Deterministic-Raise-Level-To-1}$(w)$;\\
			\textbf{Delete-From-F-List}$(u)$; \textbf{Delete-From-F-List}$(w)$; \\
		}\Else{
		\textbf{Randomised-Raise-Level-To-1}$(w)$;\\
		\uIf{$u$ is free}{
			\textbf{Naive-Settle-Augmented}$(u,1)$;\\
		}
		
	}
}\Else{
		\textbf{Delete-From-F-List}$(u)$; \textbf{Delete-From-F-List}$(w)$; \\
		
}

}

}\Else{
\For{every $x \in $ N(u)}{
	$z \leftarrow$ \textbf{Check-3-Aug-Path}$(u,x)$;\\
	\uIf{$z \neq NULL$}{
		\uIf{flag == 0}{
			\textbf{Fix-3-Aug-Path}$(u,x,mate(x),z)$;\\
			\textbf{break};
		}\Else{
			\textbf{Fix-3-Aug-Path-D}$(u,x,mate(x),z)$;\\
			\textbf{break};
		}
    }	
}
\uIf{$u$ is free}{
	\textbf{Insert-To-F-List}$(u)$;
}
}
\caption{Naive-Settle-Augmented($u$,flag) : $u$ is free, level of $u$ is $0$ and $flag$ is $0$ or $1$. Line 3 fixes Invariant~\ref{invariant_1b} for $u$.}
\label{alg:Naive-Settle-Augmented}
\end{algorithm}
\begin{algorithm}[h]
	Select an edge say $(u,y)$ uniformly at random from $O_u$;\\
	\textbf{Transfer-Ownership-To}$(y)$;\\ 
	\uIf{$y$ is matched}{
		$x \leftarrow$ mate(y);\\
		$M \leftarrow M \setminus \{(x,y)\}$; 
	}\Else{
	$x \leftarrow$ NULL;
}
$M \leftarrow M \cup \{(u,y)\}$;\\
Level$(u) \leftarrow 1$;\\
Level$(y) \leftarrow 1$;\\
\textbf{Delete-From-F-List}$(u)$; \textbf{Delete-From-F-List}$(y)$; \\

\uIf{\textbf{has-free}(u)}{
	$w \leftarrow$ \textbf{get-free}$(u)$;\\
	$z \leftarrow$\textbf{Check-3-Aug-Path}$(w,u)$;\\
	\uIf{$z \neq NULL$}{
		\textbf{Fix-3-Aug-Path-D}$(w,u,y,z)$;\\
	}	
}
return $x$;
\caption{Random-Settle-Augmented($u$) : $u$ is free, level of $u$ is $0$ and $|O_u| \geq \sqrt{n}$. Line 10 fixes Invariant~\ref{invariant 2} for $u$.}

\label{alg:Random-Settle-Augmented}
\end{algorithm}
\begin{algorithm}[h]
	$v \leftarrow$ mate$(u)$;\\
	\textbf{Take-Ownership}$(u)$;\\	
	\textbf{Transfer-Ownership-To}$(v)$;\\
	Level$(u) \leftarrow 1$;\\
	Level$(v) \leftarrow 1$;\\ 
	\caption{Deterministic-Raise-Level-To-1($u$) : $u$ is matched, level of $u$ is $0$ and $deg(u)\geq \sqrt{n}$. Line 4 fixes Invariant~\ref{invariant 3} for $u$.}
	\label{alg:Deterministic-Raise-Level-To-1}
\end{algorithm}
\begin{algorithm}[h]
	
	$v \leftarrow$ mate$(u)$;\\
	$M \leftarrow M \setminus \{u,v\}$;\\
	
	\textbf{Take-Ownership}$(u)$;\\
	$x \leftarrow $ \textbf{Random-Settle-Augmented}$(u)$;\\
	\uIf{$x \neq$ NULL}{
		\uIf{Level$(x)$ = 1}{
			\textbf{Handle-Delete-Level1}$(x,1)$;
		}\Else{
		\textbf{Naive-Settle-Augmented}$(x,1)$; 
		
	}
	
}
\uIf{$v$ is free}{
	\textbf{Naive-Settle-Augmented}$(v,1)$;
}

\caption{Randomised-Raise-Level-To-1($u$) : $u$ is matched, level of $u$ is $0$ and $deg(u)\geq \sqrt{n}$. Line 2 fixes Invariant~\ref{invariant 3} for $u$.}
\label{alg:Randomised-Raise-Level-To-1}
\end{algorithm}
\begin{algorithm}[]
\For{$p \in \{u,z\}$}{
		\textbf{Transfer-Ownership-To}$(p)$;\\
		\textbf{Delete-From-F-List}$(p)$;
}
\uIf{Level$(v)$ == 0}{
	\For{$p \in \{v,y\}$}{
		\For{$w \in N(p)$}{
					\textbf{Transfer-Ownership-To}$(p)$;\\
		}	
	}
	$Level(v) \leftarrow 1$;\\
	$Level(y) \leftarrow 1$;\\
			
}
	$M \leftarrow M \setminus \{(v,y)\}$;
	$M \leftarrow M \cup \{(u,v)\}$;
	$M \leftarrow M \cup \{(y,z)\}$;\\	
	$Level(u) \leftarrow 1$;\\
	$Level(z) \leftarrow 1$;\\
\caption{Fix-3-Aug-Path-D$(u,v,y,z)$ : $u$ is free, level of $u$ is $0$, $v$ is matched, $y$ is mate of $v$ and $z \in F(y)$ ($z \neq u$). Since path $u-v-y-z$ exists, Invariant~\ref{invariant 5} is violated. Line 13 fixes Invariant~\ref{invariant 5}. }
\label{alg:Fix-3-Aug-Path-D}	
\end{algorithm}
\begin{algorithm}[!h]
\uIf{Level$(v)$ == 1}{
			$M \leftarrow M \setminus \{(v,y)\}$;
			$M \leftarrow M \cup \{(u,v)\}$;
			$M \leftarrow M \cup \{(y,z)\}$;\\
			
			\uIf{$deg(u) \geq \sqrt{n}$}{
				\textbf{Randomised-Raise-Level-To-1}$(u)$;\\		
				\uIf{$(u,v)$ is not in matching}{
					\textbf{Handle-Delete-Level1}$(v,1)$;\\
				}
				\textbf{Transfer-Ownership-To}$(z)$;\\
				\textbf{Delete-From-F-List}$(z)$;\\
				Level$(z) \leftarrow 1$;\\
			}
			\Else{
				\uIf{$deg(z) \geq \sqrt{n}$}{
				  	\textbf{Randomised-Raise-Level-To-1}$(z)$;\\
				  	\uIf{$(y,z)$ is not in matching}{
					  	\textbf{Handle-Delete-Level1}$(y,1)$;\\		
				  	}
				  	\textbf{Transfer-Ownership-To}$(u)$;\\
				  	\textbf{Delete-From-F-List}$(u)$;\\
				  	Level$(u) \leftarrow 1$;\\
				}\Else{
					\For{$p \in \{u,z\}$}{
						
						\textbf{Transfer-Ownership-To}$(p)$;\\
						\textbf{Delete-From-F-List}$(p)$;\\
					}					
				    Level$(u) \leftarrow 1$;
					Level$(z) \leftarrow 1$;\\	
			}
		 }					
		}\Else{
		$M \leftarrow M \setminus \{(v,y)\}$;
		$M \leftarrow M \cup \{(u,v)\}$;
		$M \leftarrow M \cup \{(y,z)\}$;\\
		\uIf{$deg(u) \geq \sqrt{n}$}{
			\textbf{Randomised-Raise-Level-To-1}$(u)$;\\
			\uIf{v is free}{
				\textbf{Naive-Settle-Augmented}$(v,1)$;
			}
		}
		\uIf{${deg(z) \geq \sqrt{n}}$}{
			\textbf{Randomised-Raise-Level-To-1}$(z)$;\\
			\uIf{y is free}{
				\textbf{Naive-Settle-Augmented}$(y,1)$;
				}
		}
		\uIf{${deg(u) < \sqrt{n}}$}{
			\textbf{Delete-From-F-List}$(u)$;\\			
		}
		\uIf{${deg(z) < \sqrt{n}}$}{
			\textbf{Delete-From-F-List}$(z)$;\\			
		}
}		
	
\caption{Fix-3-Aug-Path($u,v,y,z$) : $u$ is free, level of $u$ is $0$, $v$ is matched, $y$ is mate of $v$ and $z \in F(y)$ ($z \neq u$). Since path $u-v-y-z$ exists, Invariant~\ref{invariant 5} is violated. Line 2 and Line 27  fixes Invariant~\ref{invariant 5}. }
\label{alg:Fix-3-Aug-Path}

\end{algorithm}
\begin{algorithm}[]
	\textbf{Transfer-Ownership-From}$(u)$\\
	Level$(u) \leftarrow 0$;\\	
	\uIf{$|O_u| \geq \sqrt{n}$}{
		$x \leftarrow$ \textbf{Random-Settle-Augmented$(u)$};\\
		\uIf{$x \neq$ NULL}{
			\textbf{Naive-Settle-Augmented}$(x,1)$;
		}
	}\Else{
	\textbf{Naive-Settle-Augmented}$(u,flag)$;\\		
}	
\caption{Handle-Delete-Level1($u$,$flag$) : $u$ is free, level of $u$ is $1$ and $flag$ is $0$ or $1$. Line 2 fixes Invariant~\ref{invariant_1a}a for $u$.}
\label{alg:Handle-Delete-Level1}
\end{algorithm}
\begin{algorithm}[]
	\uIf{$|O_u| \geq |O_v|$}{	
		Add $(u,v)$ to $O_u$;\\
	}\Else{	
		Add $(u,v)$ to $O_v$;\\
	}	
	\uIf{both $u$ and $v$ are free}{
		$M \leftarrow M \cup \{(u,v)\}$;
		Flag-uv-matched $\leftarrow 1$;\\
	}
	\uIf{$|O_v| > |O_u|$}{
		Swap $u$ and $v$ for remaining processing;
	}
	\uIf{$|O_u|$ == $\sqrt{n}$}{
		\textbf{Transfer-Ownership-To}$(u)$;\\
		\uIf{$u$ is matched}{
			$M \leftarrow M \setminus \{(u,mate(u))\}$;
		}
		$x \leftarrow$ \textbf{Random-Settle-Augmented}($u$);\\
		\uIf{$x \neq$ NULL}{
			\textbf{Naive-Settle-Augmented}($x,1$);\\
		}
		\uIf{Flag-uv-matched == 1}{
			\uIf{mate$(u) \neq v$ AND $v$ is free}{
				\textbf{Naive-Settle-Augmented}($v,1$);\\
			}
		}\Else{
			\uIf{$v$ is matched}{
				\uIf{$deg(v) \geq \sqrt{n}$ AND Level$(v)$ == $0$}{
					\textbf{Deterministic-Raise-Level-To-1}$(v)$;	
				}
			}	
		}
	}\Else{
	\uIf{$v$ is matched}{
		\uIf{$deg(v) \geq \sqrt{n}$}{
			\textbf{Randomised-Raise-Level-To-1}$(v)$;\\
			\uIf{$u$ is matched AND $deg(u) \geq \sqrt{n}$ AND Level$(u)$ == $0$}{
					\textbf{Deterministic-Raise-Level-To-1}$(u)$;	
				
			}		
			
		}\Else{
		\uIf{$u$ is free}{
			$z \leftarrow$ \textbf{Check-3-Aug-Path}$(u,v)$;\\
			\uIf{$z \neq NULL$}{
				\textbf{Fix-3-Aug-Path}($u,v,mate(v),z$);
			}
		}\Else{
		\uIf{$deg(u) \geq \sqrt{n}$}{
			\textbf{Randomised-Raise-Level-To-1}$(u)$;	
		}
		
	}
	
}	
}\Else{
\uIf{$u$ is matched}{
	\uIf{$deg(u) \geq \sqrt{n}$}{
		\textbf{Randomised-Raise-Level-To-1}$(u)$;	
		
	}\Else{
	\uIf{$v$ is free}{
		$z \leftarrow$ \textbf{Check-3-Aug-Path}$(v,u)$;\\
		\uIf{$z \neq NULL$}{
			\textbf{Fix-3-Aug-Path}($v,u,mate(u),z$);
		}
	}
}

}
}
\uIf{Flag-uv-matched == 1 AND $deg(u) < \sqrt{n}$ AND $deg(v) < \sqrt{n}$ }{
		\textbf{Delete-From-F-List}$(u)$; \textbf{Delete-From-F-List}$(v)$;\\	
		
}

}
\caption{Handle-Insert-Level0($u,v$)}
\label{alg:Insertion-Level0}
\end{algorithm}

\begin{algorithm}[h]
	Insert $v$ to $N(u)$ and $u$ to $N(v)$;\\
	\uIf{Level$(u)$==1 and Level$(v)$==1}
	{	
		Include $(u,v)$ suitably into $O_u$ or $O_v$;\\  
	}\ElseIf{Level$(u)$==1 and Level$(v)$==0}
	{
		Add $(u,v)$ to $O_u$ ;\\
		\uIf{$v$ is free}{
			$z \leftarrow$ \textbf{Check-3-Aug-Path}$(v,u)$;\\
			\uIf{$z \neq NULL$}{
				\textbf{Fix-3-Aug-Path}($v,u,mate(u),z$);
			}	
		}\Else{
		\uIf{$deg(v) \geq \sqrt{n}$}{
			\textbf{Randomised-Raise-Level-To-1}$(v)$
		}
	}
	
}\ElseIf{Level(u)==0 and Level(v)==1}
{
	Add $(u,v)$ to $O_v$ ;\\
	\uIf{$u$ is free}{
		$z \leftarrow$ \textbf{Check-3-Aug-Path}$(u,v)$;\\
		\uIf{$z \neq NULL$}{
			\textbf{Fix-3-Aug-Path}($u,v,mate(v),z$);
		}		
	}\Else{
	\uIf{$deg(u) \geq \sqrt{n}$}{
		\textbf{Randomised-Raise-Level-To-1}$(u)$
	}
}

}\Else{
Handle-Insertion-Level0($u,v$);
}
\caption{Insert($u,v$)}
\label{alg:Insert}	
\end{algorithm}
\begin{algorithm}[h]
	Adjust $O_u$ and $O_v$ suitably;\\	
	Delete $v$ from $N(u)$ and $u$ from $N(v)$;\\ 
	\uIf{$(u,v)$ is unmatched}{
		return;
	}\Else{
	$M \leftarrow M \setminus \{(u,v)\}$;\\
	\uIf{Level$(u,v)$ == 0}{
		\textbf{Naive-Settle-Augmented}$(u,0)$;\\
		\textbf{Naive-Settle-Augmented}$(v,0)$;\\
	}\Else{
	\textbf{Handle-Delete-Level1}$(u,0)$;\\
	\textbf{Handle-Delete-Level1}$(v,0)$;
}
}
\caption{Delete($u,v$)}
\label{alg:Deletion}
\end{algorithm}



\vspace{.5cm} 
\end{document}